\documentclass[12pt]{article}
\usepackage{amsmath,amssymb,epsf,epsfig,amsthm,times,ifthen,epic,psfrag}
\usepackage{enumerate,multirow,makecell,tikz,pgfplots}
\usepackage{cite}
\usepackage{fancyhdr}
\usepackage{setspace}
\usepackage{lastpage}
\usepackage{url}
\usepackage{hyperref}
\usepackage{refcount}

\def\submitteddate{June 4, 2017}
\def\reviseddate{January 29, 2018}

\renewcommand{\baselinestretch}{1}
\DeclareFontFamily{U}{matha}{\hyphenchar\font45}
\DeclareFontShape{U}{matha}{m}{n}{
      <5> <6> <7> <8> <9> <10> gen * matha
      <10.95> matha10 <12> <14.4> <17.28> <20.74> <24.88> matha12
      }{}
\DeclareSymbolFont{matha}{U}{matha}{m}{n}
\DeclareMathSymbol{\notdivides}{3}{matha}{"1F}
\DeclareMathSymbol{\divides}{3}{matha}{"17}

\begin{document}

\newcommand{\creationtime}{\today}

\pagestyle{fancy}
\renewcommand{\headrulewidth}{0cm}
\chead{\footnotesize{Connelly-Zeger}}
\rhead{\footnotesize{\reviseddate}}
\lhead{\footnotesize{\textit{Ring Capacity - Draft}}}
\cfoot{Page \arabic{page} of \pageref{LastPage}}

\renewcommand{\qedsymbol}{$\blacksquare$}

\newtheorem{theorem}              {Theorem}     [section]
\newtheorem{lemma}      [theorem] {Lemma}
\newtheorem{corollary}  [theorem] {Corollary}
\newtheorem{proposition}[theorem] {Proposition}
\newtheorem{remark}     [theorem] {Remark}
\newtheorem{conjecture} [theorem] {Conjecture}
\newtheorem{example}    [theorem] {Example}

\theoremstyle{definition}         
\newtheorem{definition} [theorem] {Definition}
\newtheorem*{claim} {Claim}
\newtheorem*{notation}  {Notation}

\newcommand{\CommHomomorphismLemma}{Lemma II.5}             
\newcommand{\CommLemmaFieldNoDominance}{Lemma III.2}        
\newcommand{\NonModuleHomomorphism}{Lemma I.6}              
\newcommand{\NonCommSimpleHomoLemmas}{Lemmas II.1 and II.3} 
\newcommand{\NonFaithfulModule}{Lemma II.6}                 
\newcommand{\NonSmallerModule}{Lemma II.9}                  

\newcommand{\Z}{\mathbb{Z}}
\newcommand{\Q}{\mathbb{Q}}

\newcommand{\alphabet}{\mathcal{A}}
\newcommand{\A}{\mathcal{A}}
\newcommand{\F}{\mathbb{F}}
\newcommand{\K}{\mathbb{K}}
\newcommand{\UniLinCap}[2]{\mathcal{C}_{lin}(#1,#2)}
\newcommand{\UniCapacity}[1]{\mathcal{C}(#1)}
\newcommand{\LinRegion}[2]{\mathcal{R}_{lin}(#1,#2)}
\newcommand{\Region}[1]{\mathcal{R}(#1)}

\newcommand{\rateregion}{rate region}
\newcommand{\rateregions}{rate regions}
\newcommand{\linearrateregion}{linear rate region}
\newcommand{\linearrateregions}{linear rate regions}

\newcommand{\ith}[2]{#1^{(#2)}}
\renewcommand{\vec}[2]{\left[#1 \right]_{#2}}

\newcommand{\scaDom}{scalarly dominates}
\newcommand{\fracDom}{fractionally dominates}
\newcommand{\nfracDom}{does not fractionally dominate}
\newcommand{\fracdominated}{fractionally dominated}
\newcommand{\fractionaldominance}{fractional dominance}
\newcommand{\scalardominance}{scalar dominance}

\newcommand{\charNetwork}[1]{Char-#1 network}

\newcommand{\Module}[2]{{}_{#2} #1}

\newcommand{\GF}[1]{\mathrm{GF}\!\left(#1\right)}
\newcommand{\Char}[1]{\mathsf{char}\!\left(#1\right)}
\newcommand{\Comment}[1]{ \left[\mbox{from  #1} \right]}
\newcommand{\Div}[2]{ #1  \bigm|  #2  } 
\newcommand{\NDiv}[2]{ #1   \notdivides   #2   }
\newcommand{\GCD}[2]{ \mathsf{gcd} \left( #1  ,  #2 \right)  } 
\newcommand{\PrimeFact}[1]{p_1^{k_1} \cdots p_{#1}^{k_{#1}}}

\newcommand{\osum}{\displaystyle\bigoplus}
\newcommand{\dsum}{\displaystyle\sum}
\newcommand{\act}{\cdot}
\newcommand{\newaction}{\odot}

\renewcommand{\emptyset}{\varnothing} 
\newcommand{\Network}{\mathcal{N}}
\newcommand{\knNetwork}[1]{\Network^{(#1)}}

\newcommand{\TBA}{*** To Be Added ***}

\setcounter{page}{1}

\title{Capacity and Achievable Rate Regions for Linear Network Coding over Ring Alphabets
  \thanks{This work was supported by the 
  National Science Foundation.
  \newline
  \indent \textbf{J. Connelly and K. Zeger} are with the 
  Department of Electrical and Computer Engineering, 
  University of California, San Diego, 
  La Jolla, CA 92093-0407 
  (j2connelly@ucsd.edu and zeger@ucsd.edu).
  }
}
\author{Joseph Connelly and Kenneth Zeger\\}
\date{
  \textit{
  IEEE Transactions on Information Theory\\
  Submitted: \submitteddate\\
  Revised: \reviseddate
  }
}

\maketitle
\begin{abstract}
  The rate of a network code is the ratio
  of the block size of the network's messages to that of its edge codewords.
  We compare the linear capacities and achievable rate regions of networks
  using finite field alphabets to the more general cases of arbitrary ring and module alphabets.
  For non-commutative rings, two-sided linearity is allowed.
  Specifically, we prove the following for directed acyclic networks:
  \begin{itemize}
    \item[(i)] The \linearrateregion{} and the linear capacity
    of any network over a finite field
    depend only on the characteristic of the field.
    Furthermore, any two fields with different characteristics
    yield different linear capacities for at least one network.

    \item[(ii)] Whenever the characteristic of a given finite field divides the size of a given finite ring,
    each network's \linearrateregion{} over the ring
    is contained in its \linearrateregion{} over the field.
    Thus, any network's linear capacity over a field is at least its linear capacity over any other ring of the same size.
    An analogous result also holds for linear network codes over module alphabets.

    \item[(iii)] Whenever the characteristic of a given finite field does not divide the size of a given finite ring,
    there is some network whose linear capacity over the ring
    is strictly greater than its linear capacity over the field.
    Thus, for any finite field, 
    there always exist rings over which some networks
    have higher linear capacities than over the field.
  \end{itemize}
\end{abstract}
\clearpage
\section{Introduction} 
  \label{sec:intro}

In network coding, solvability determines whether or not a network's receivers
can adequately deduce from their inputs a specified subset of the network's message values.
The solvability of directed acyclic networks
follows a hierarchy of different types of network coding.
For example, scalar linear coding over finite fields is known to be 
inferior to vector linear coding over finite fields \cite{Medard-NonMulticast},
which in turn is known to be inferior to non-linear coding \cite{DFZ-Insufficiency}.
On the other hand, the capacity of a network reveals how much transmitted information
per channel use 
(i.e., source messages per edge use)
can be sent to the network's receiver nodes
in the limit of large block sizes for transmission.
It is also known that
linear codes over finite fields
cannot achieve the full capacity of some networks \cite{DFZ-Insufficiency}.
Thus,
linear coding over finite fields
is inferior to more general types of network coding
in terms of both solvability and capacity.
Nevertheless, linear codes over finite fields
are attractive for
both theoretical and practical reasons
\cite{Li-LNC}.

In certain cases, 
linear coding over finite ring alphabets
can offer solvability advantages
over finite field alphabets \cite{CZ-Commutative,CZ-NonCommutative}.
An open question has been whether the linear capacity of a network 
over a finite field can be improved by using some other ring of
the same size as the field.
In other words, does the improvement in network solvability,
from using more general rings than fields,
also carry over to network capacity?
In the present paper, we answer this question in the negative.
That is, we prove that the linear capacity of a network cannot be improved by changing
the network coding alphabet from a field to any other ring of the same size.

Another open question has been whether the linear capacity of a network
over a finite field can depend on any aspect of the field other than its
characteristic.
Indeed it has been previously observed that the linear capacity of a network
can vary as a function of the field
(e.g., \cite{CZ-NonLinear, DFZ-Regions, DFZ-CharacteristicInequalities}),
but all known examples had linear capacities that only depended on
the fields' characteristics.
We also answer this question in the negative.
That is, 
we prove that any two fields with the same characteristic
will result in the same linear capacity for any given network.
Furthermore, 
any two fields with different characteristics will result in 
different linear capacities for at least one network.
We prove analogous (and more general) results for linearly achievable rate regions of networks
over finite fields.

Unlike finite fields, 
finite rings need not have prime-power size,
which may be advantageous in certain applications.
An open question has been
whether a network can increase its linearly achievable rate region
by allowing the alphabet to be a ring of non-power-of-prime size.
However, we again answer this question in the negative
by showing that 
a network's \linearrateregion{} over a ring
is contained in its \linearrateregion{} over any field
whose characteristic divides the ring's size.
This result follows from the fact that
every finite ring is isomorphic to some direct product of rings of prime-power sizes.
As a consequence of this result,
any network's linear capacity over a particular ring
is at most its linear capacity over any field
whose characteristic divides the ring's size.
These results extend naturally to the more general case
of linear network codes in which the alphabet has the structure of a finite module.

\newpage

\subsection{Modules, Linear Functions, and Tensor Products}
\label{ssec:modules}

We focus on linear network codes over finite rings,
but we prove many of our intermediate results in the broader context of linear network codes over modules.
In this section, we define linear functions over modules,
which generalize linear functions over rings.
We then formally define linear network codes over rings and modules in Section~\ref{ssec:ring-linearity}.

\begin{definition}
  A \textit{left $R$-module} is
  an Abelian group
  $(G,\oplus)$ together with a ring 
  $(R,+,*)$ of \textit{scalars}
  and an action 
  $\cdot : R \times G \to G $
  such that
  for all $r,s \in R$ and all $g,h \in G$ the following hold:
  \begin{align*}
    r \cdot (g\oplus h) &= (r \cdot g) \oplus (r \cdot h) \\
    (r + s) \cdot g &= (r \cdot g) \oplus (s \cdot g) \\
    (r*s) \cdot g &= r \cdot (s \cdot g) \\
    1 \cdot g &= g .
  \end{align*}
  \label{def:module}
  \vspace{-2.5em}
\end{definition}
From these properties, it also follows that
$0 \act g = 0$ and $r \act 0 = 0$
for all $g \in G$ and all $r \in R$.
For brevity, we will sometimes refer to such an $R$-module 
as $\Module{G}{R}$ or simply the $R$-module $G$.
Since network coding alphabets are presumed to be finite,
a module will always refer to a module in which $G$ is finite.
However, in principle,
the ring need not be finite,
so we make no assumptions about the cardinality of the ring in a module.
Some important examples of modules include:
\begin{itemize}\itemsep0em 
\vspace{-0.5em}
  \item The ring of integers $\Z$ acts on any Abelian group $G$ by repeated addition in $G$.

  \item Any ring $R$ acts on its own additive group $(R,+)$
  by multiplication in $R$.
  We denote this module by $\Module{R}{R}$.

  \item Any ring $R$ acts on the set of all $t$-vectors over $R$, denoted by $R^t$, 
  by scalar multiplication.
  When $R$ is a field, this module is a vector space.

  \item If $\Module{G}{R}$ is a module,
  then the ring of all $t \times t$ matrices with entries in $R$, denoted $M_t(R)$,
  acts on the group, $G^t$, of all $t$-vectors over $G$
  via matrix-vector multiplication
  where multiplication of elements of $R$ with elements of $G$ is given by the action of $\Module{G}{R}$.
  A special case of this module, $\Module{G^t}{M_t(R)}$, occurs when $G = (R,+)$,
  in which case matrices over $R$
  act on vectors over $R$
  via matrix-vector multiplication over $R$.
\end{itemize}

If $R$ is a ring,
a function $f: R^m \to R$ of the form
$$f(x_1,\dots,x_m) = a_1 \, x_1 + \cdots + a_m \, x_m$$
where $a_1,\dots,a_m \in R$,
is a (left) \textit{one-sided linear function}
with respect to both the ring $R$
and the left module $\Module{R}{R}$.%
\footnote{
Every right one-sided linear function with respect to a ring or a right module
can be written as 
a corresponding left one-sided linear function with respect to
a left module with the same Abelian group.
Hence, in this paper, it suffices for us to exclusively use left one-sided linear functions.
}
A function $f': R^m \to R$ of the form
\begin{align}
  f'(x_1,\dots,x_m) = \sum_{i=1}^m \sum_{j=1}^{n_i} a_{i,j} \, x_i \, b_{i,j} \label{eq:two-sided-1}
\end{align}
where $a_{i,j}, b_{i,j} \in R$,
is a \textit{two-sided linear function}
with respect to $R$.
When $R$ is commutative, every two-sided linear function
is also a one-sided linear function,
since in a commutative ring,
$$\sum_{i=1}^m \sum_{j=1}^{n_i} a_{i,j} \, x_i \, b_{i,j} = \sum_{i=1}^m \left( \sum_{j=1}^{n_i} a_{i,j} \, b_{i,j} \right) \, x_{i}.$$
However, 
left and right multiplication are not necessarily the same in a non-commutative ring,
so the class of two-sided linear functions is broader than the class of one-sided linear functions.

\begin{example}
Let $R$ be the (non-commutative) ring of all $2 \times 2$ matrices over a field.
The function $f: R \to R$ given by
\begin{align*}
f\left( \left[\begin{array}{cc}
  x_{1,1} & x_{1,2} \\
  x_{2,1} & x_{2,2} 
\end{array} \right] \right)
&=
\left[\begin{array}{cc}
  1 & 0 \\
  0 & 0
\end{array} \right]
\,
\left[\begin{array}{cc}
  x_{1,1} & x_{1,2} \\
  x_{2,1} & x_{2,2} 
\end{array} \right] 
\,
\left[\begin{array}{cc}
  1 & 0 \\
  0 & 0
\end{array} \right]
+
\left[\begin{array}{cc}
  0 & 0 \\
  0 & 1
\end{array} \right]
\,
\left[\begin{array}{cc}
  x_{1,1} & x_{1,2} \\
  x_{2,1} & x_{2,2} 
\end{array} \right] 
\,
\left[\begin{array}{cc}
  0 & 0\\
  0 & 1
\end{array} \right]
\\
&= \left[\begin{array}{cc}
  x_{1,1} & 0 \\
  0 & x_{2,2} 
\end{array} \right]
\end{align*}
is a two-sided linear function over $R$.
It can be verified that, for all $A,B \in R$,
the function $f(X)$ is not the function $A X B$.
By allowing for sums of $X$ terms multiplied by coefficients on both the left and the right,
a broader class of functions can be attained
than with a single $X$ term multiplied by coefficients on the left and the right.
This also implies $f(X)$ cannot be written as a (left or right) one-sided linear function.
\label{ex:two-sided}
\end{example}

In the remainder of this section,
we will show that two-sided linear functions over rings
can be written as one-sided linear functions with respect to some module,
i.e., $f'$ in \eqref{eq:two-sided-1} can be written as
$$f'(x_1,\dots,x_m) =  c_1 \act x_1 + \cdots + c_m \act x_m$$
where $c_1,\dots,c_m$ are elements of some other ring
that acts on $R$.
In order to do so, we exploit module tensor products.
If $G$ and $H$ are each $R$-modules,
then the tensor product of $G$ and $H$
is a third $R$-module
that satisfies properties similar 
to the constructed vector space in the following example.

\begin{example}
Suppose $\F$ is a field
and $U \subseteq \F^m $ and $V \subseteq \F^n$ are vector spaces.
For each $u \in U$ and $v \in V$,
define the $mn$ vector $(u,v)$ by
$$ (u,v) =
\left[ \begin{array}{c}
  u_1 v_1 \\
  \vdots \\
  u_1 v_n \\
  \vdots \\
  u_m v_1 \\
  \vdots \\
  u_m v_n
\end{array} \right].$$
It is easily verified that
for all $u,u' \in U$, all $v,v' \in V$, and all $\alpha \in \F$,
\begin{align*}
  (u, v) + (u', v) &= (u + u', \, v) \\
  (u, v) + (u, v') &= (u, \, v + v') \\
    \alpha \, (u,v) &= (\alpha u, \, v)\\
   \alpha \, (u,v) &= (u, \, \alpha v).
\end{align*}
The subspace of $\F^{mn}$ generated by all vectors of the form
$(u,v)$ for some $u \in U$ and some $v \in V$
is isomorphic to the tensor product of $U$ and $V$.
In general, this tensor product space
differs from the direct product space
$U \times V\subseteq \F^{m+n}$ 
obtained by concatenating 
vectors from $U$ with vectors from $V$.
In fact, when $U = \F^m$ and $V = \F^n$,
the tensor product space is $\F^{mn}$,
whereas the direct product space is $\F^{m+n}$.
\label{ex:vector-product}
\end{example}

If $R$ is a ring
and $E$ is a set,
the \textit{free $R$-module generated by $E$}
is denoted $R^{(E)}$.
In this module,
the group is
the subset of the Cartesian product 
$\displaystyle\prod_{e \in E} R$ 
consisting only of the elements that have finitely many non-zero components together with component-wise addition,
and the ring $R$ acts on $R^{(E)}$ component-wise.
By mapping the element $e \in E$
to the vector in $R^{(E)}$
whose $e$th component is $1$
and all other components are $0$,
we can view $R^{(E)}$ as the set of all 
finite $R$-linear combinations of elements of $E$.
In other words, 
every element of $R^{(E)}$ can be uniquely written as
$\dsum_{e \in E} a_e \, e$,
where only finitely many $a_e \in R$ are non-zero,
so the set $E$
is a basis for $R^{(E)}$.

If $G$ is an $R$-module
and $N$ is a subgroup of $G$ that is closed under the action of $R$,
then $N$ is a \textit{submodule} of $G$.
The quotient group $G/N$ is
also an $R$-module
(e.g., see \cite[p. 348]{Dummit-Algebra}).
If $E$ is a subset of $G$,
then the \textit{submodule generated by $E$}
is 
$$
\{r_1 e_1 + \cdots + r_m e_m \; : \; m \in \mathbb{N}, \, r_1,\dots,r_m \in R, \, e_1,\dots,e_m \in E\}.
$$

Now let $R$ be a commutative ring,
let $G$ and $H$ be $R$-modules,
and let $N$ be the submodule of $R^{(G \times H)}$ generated by the set
$$\left\{ \begin{array}{l}
  (g,h) + (g',h) - (g+g',h), \\
         (g,h') + (g,h) - (g,h+h'), \\
         r \, (g,h) - (r  g, h), \\
         r \, (g,h) - (g, r \, h) 
    \end{array}
         \; : \;
         g,g' \in G, \, h,h' \in H, \, r \in R
    \right\}.$$
The \textit{tensor product module of $\Module{G}{R}$ and $\Module{H}{R}$},
denoted $G \otimes_R H$,
is the quotient $R$-module $R^{(E)} / N$.
In other words,
$G \otimes_R H$ is the set of equivalence classes
of the congruence generated by the following relations on $R^{(E)}$:
\begin{align*}
   (g,h) + (g',h) & = (g+g', h) \\
   (g,h) + (g,h') & = (g, h+h') \\
    r\, (g,h) &= (r \, g,h) \\
   r\, (g,h) &= (g,r \, h).
\end{align*}
The tensor product module is unique up to isomorphism
(e.g., see \cite[Sections 10.1 -- 10.4]{Dummit-Algebra} for more information on modules and tensor products) and
exhibits similar properties to the tensor product of vector spaces.
The elements of $G \otimes_R H$ are called \textit{tensors}
and can be written (non-uniquely, in general)
as sums of equivalence class representatives:
$(g_1,h_1) + \cdots + (g_m, h_m)$,
for some positive integer $m$ and 
$(g_1,h_1),\dots,(g_m,h_m) \in G \times H$.

\begin{definition}
Let $R$ and $S$ be finite rings, and let $\Z$ denote the ring of integers.
The \textit{tensor product ring $R \otimes S$}
is the Abelian group $R \otimes_{\Z} S$
together with multiplication given by
$$\left( \sum_{i=1}^m  (r_i, s_i) \right)
* \left( \sum_{i=j}^n  (r_j', s_j') \right)
= \sum_{i=1}^m \sum_{j=1}^n (r_i r_j', \, s_i s_j')$$
for all $\big( \sum_{i=1}^m (r_i, s_i) \big),
\big( \sum_{i=j}^n (r_j', s_j') \big) 
\in R \otimes_{\Z} S$.
\end{definition}
This tensor product ring is well defined
and unique up to isomorphism
(e.g., see \cite[Chapter 10.4, Proposition 21]{Dummit-Algebra}).
As an example,
if $\Z_m$ and $\Z_n$ denote the rings of integers modulo $m$ and $n$, respectively,
then
we have 
$\Z_m \otimes \Z_n \cong \Z_{\GCD{m}{n}}$
(e.g., see \cite[p. 369]{Dummit-Algebra}).
Specifically, if $m = 4$ and $n = 2$, then the tensors in $\Z_4 \otimes \Z_2$ are such that
$$(0,0) = (0,1) = (2,1) = (1,0) = (2,0) = (3,0) 
 \; \; \text{ and } \; \; (1,1) = (3,1)$$
and addition and multiplication are isomorphic to addition and multiplication in $\Z_2$.

We also comment that the \textit{direct product ring}
$R \times S$ with component-wise addition and multiplication 
is generally not isomorphic to the tensor product ring $R \otimes S$.
As an example, 
if $m$ and $n$ are relatively prime,
then by the Chinese remainder theorem, $\Z_m \times \Z_n \cong \Z_{mn}$
(e.g., see \cite[p. 267]{Dummit-Algebra}),
whereas $\Z_{m} \otimes \Z_n \cong \Z_1$ is the trivial ring.

For a finite ring $R$,
the \textit{opposite ring},
denoted $R^{op}$, is the additive group of $R$
with multiplication taken in the opposite order,
i.e., $a *_{op} b = b a$, for all $a,b \in R$.
The tensor product ring
$R \otimes R^{op}$
acts on $(R,+)$ via
$$\left(\sum_{i=1}^n (a_i, b_i) \right) \cdot r = \sum_{i=1}^n a_i \, r \, b_i$$
for all $a_1,\dots,a_n,b_1,\dots,b_n,r \in R$.
In other words, $R\otimes R^{op}$
acts on $(R,+)$ by computing \textit{two-sided} linear combinations of elements of $(R,+)$.
We denote this module by
$\Module{R}{R \otimes R^{op}}$.
The properties of tensor addition and multiplication are natural in the context of this module.
In particular, for all $a,a',b,b',x \in R$, and $n \in \Z$, we have
\begin{align*}
  \left( (a, b) + (a', b) \right) \act x 
    &= a \, x \, b + a' \, x \, b
    = (a + a') \, x \, b 
    = (a+a', \, b) \act x
    \\
  \left( (a, b) + (a, b') \right) \act x 
    &= a \, x \, b + a \, x \, b'
    = a \, x \, (b + b') 
    = (a, \, b+b') \act x
    \\
  n \, (a,b) \act x 
    &= n \, (a \, x \, b)  
    = (n \, a) \, x \, b
    = (n a, b) \act x 
    \\
  n \, (a,b) \act x 
    &= n \, (a \, x \, b)  
    = a \, x \, (n \, b)
    = (a, n b) \act x .
\end{align*}

The two-sided linear function $f'$ in \eqref{eq:two-sided-1}
can now be written as
$$f'(x_1,\dots,x_m) = \sum_{i=1}^m \left( \sum_{j=1}^{n_i} (a_{i,j} , b_{i,j}) \right) \act x_i$$
which is a one-sided linear function
with respect to the $R \otimes R^{op}$-module $R$.
This shows that \textit{one-sided} linearity over left modules
generalizes \textit{two-sided} linearity over rings.

\begin{example}
Let $R$ be the (non-commutative) ring of all $2 \times 2$ matrices over a field.
The two-sided linear function $f:R \to R$
from Example~\ref{ex:two-sided}
can be written as a one-sided linear function over 
the $R \otimes R^{op}$-module $R$ as
\begin{align*}
f\left( \left[\begin{array}{cc}
  x_{1,1} & x_{1,2} \\
  x_{2,1} & x_{2,2} 
\end{array} \right] \right)
&=
\left(
\left(\left[\begin{array}{cc}
  1 & 0 \\
  0 & 0
\end{array} \right]
  , \;
\left[\begin{array}{cc}
  1 & 0 \\
  0 & 0
\end{array} \right] \right)
+
\left( \left[\begin{array}{cc}
  0 & 0 \\
  0 & 1
\end{array} \right]
  , \;
\left[\begin{array}{cc}
  0 & 0\\
  0 & 1
\end{array} \right] \right)
\right)
\act
\left[\begin{array}{cc}
  x_{1,1} & x_{1,2} \\
  x_{2,1} & x_{2,2} 
\end{array} \right]
\\
\\
&= \left[\begin{array}{cc}
  x_{1,1} & 0 \\
  0 & x_{2,2} 
\end{array} \right].
\end{align*}
\end{example}

\subsection{Network Coding Model}
  \label{sec:model}

A \textit{network} will refer to a finite, directed, acyclic multigraph,
some of whose nodes are \textit{sources} or \textit{receivers}.
Source nodes generate 
\textit{message vectors} 
whose components are arbitrary elements 
of a fixed, finite set of size at least $2$,
called an \textit{alphabet}.
The elements of an alphabet are called \textit{symbols}.
We will denote the cardinality of an alphabet $\A$ by $|\A|$.
The \textit{inputs} to a node are the message vectors, 
if any, originating at the node
and the symbols on the incoming edges of the node.
Each outgoing edge of a network node
has associated with it an \textit{edge function}
that maps the node's inputs
to the 
vector of symbols carried by the edge, 
called the \textit{edge vector}.
Each receiver node has \textit{decoding functions}
that map the receiver's inputs
to a vector of alphabet symbols in an attempt to 
recover the receiver's \textit{demands},
which are the message vectors the receiver wishes to obtain.

In a network with $m$ message vectors,
a \textit{$(k_1,\dots,k_m,n)$ code over an alphabet $\A$}
(also called a \textit{fractional} code)
is an assignment of edge functions to the edges in the network
and an assignment of decoding functions to the receivers in the network
such that the $i$th message vector is an element of $\A^{k_i}$
and the edge vectors are elements of $\A^n$.
The \textit{rate vector} of a $(k_1,\dots,k_m,n)$ network code is $\mathbf{r} = (k_1/n,\dots,k_m/n)$.
A fractional code is a \textit{solution} if
each receiver recovers its demanded message vector
from its inputs,
and a rate vector $\mathbf{r}$
is \textit{achievable for a network} 
if the network has a fractional solution
with rate vector $\mathbf{r}$
over some alphabet.

\newpage

\subsection{Linearity over Finite Rings and Modules}
\label{ssec:ring-linearity}

A function 
$f: G^{s} \to G^t$
is \textit{linear with respect to the module $\Module{G}{R}$} if
it can be written as a matrix-vector product,
$f(\mathbf{x}) =
A\mathbf{x}
$,
where 
\begin{itemize}
\itemsep0em
  \item $A$ is a $t \times s$ matrix with elements from $R$,
  \item multiplication of elements of $R$ by elements of $G$
  is the action of the module.
\end{itemize}

A fractional code is \textit{linear over the $R$-module $G$} if 
the message vectors and edge vectors have components from $G$
and all edge functions and decoding functions are linear over the module.
For each network node, 
the vector $\mathbf{x} \in G^s$ is a concatenation of all the input vectors of the node.
In other words, the network alphabet is $G$,
and the outgoing edge vectors and decoded symbol vectors 
at a node
are linear combinations of the node's vector inputs,
where the coefficients describing the linear combination are from $R$.
We use modules as a tool to prove results related to linear coding over rings,
since linear network coding over modules generalizes linear network coding over rings and fields.
The module approach is especially useful for non-commutative rings with two-sided linear codes.

If $R$ is a finite ring,
then a fractional linear code over the module $\Module{R}{R \otimes R^{op}}$
is said to be a \textit{fractional two-sided linear code over $R$}.
In particular, the network alphabet is $R$,
and the outgoing edge vectors and decoded symbols
carry linear combinations of the node's input components,
where each input component in the combination is multiplied on the left and right by constants from $R$.
If $R$ is commutative, then then a fractional two-sided linear code over $R$
is also a fractional linear code over the module $\Module{R}{R}$,
since one-sided and two-sided linearity are equivalent in this case.
A rate vector $\mathbf{r}$
is \textit{linearly achievable for a network over a finite ring $R$}
if the network has a fractional two-sided linear solution over $R$
with rate vector $\mathbf{r}$.

\subsection{Rate Regions, Capacity, and Solvability}

The \textit{\rateregion{}}
of a network $\Network$ is
$$ \Region{\Network} =\{\mathbf{r} \in \Q^m \, : \, \mathbf{r} \text{ is achievable for } \Network \},%
  \footnote{Some authors refer to the \rateregions{} and \linearrateregions{} of networks as ``capacity regions'' or ``achievable rate regions''
  and sometimes define them as the convex hull or the topological closure of the set.
  We compare a network's \linearrateregions{}
  over finite rings to its \linearrateregions{} over finite fields,
  and our results immediately extend to these alternate definitions of \rateregions.
  }$$
the \textit{capacity} (also known as the ``uniform capacity'' or the ``symmetric capacity'') 
is
$$ \UniCapacity{\Network} = \text{sup } \{ r \in \Q \, : \, (r,\dots,r) \text{ is achievable for } \Network  \} ,$$
the \textit{\linearrateregion{}}
with respect to a ring alphabet $R$ is
$$ \LinRegion{\Network}{R} = \{\mathbf{r} \in \Q^m \, : \, \mathbf{r} \text{ is linearly achievable for $\Network$ over } R \}  ,$$
and
the \textit{linear capacity}
with respect to a ring alphabet $R$ is
$$ \UniLinCap{\Network}{R} = \text{sup } \{ r \in \Q \, : \, (r,\dots,r) \text{ is linearly achievable for $\Network$ over } R \} .$$
%

While the emphasis of this paper 
is on \rateregions{} and capacities of networks,
we define several solvability properties,
as they will be useful in proving our main results.
A $(k_1,\dots,k_m,n)$ code,
for which $k_1=\dots = k_m = n = t$, 
is also called a \textit{$t$-dimensional vector} code,
i.e., the block size of every message and edge is $t$,
and a $1$-dimensional vector code is called a \textit{scalar} code.
A network is said to be 
\begin{itemize}\itemsep0em
  \item \textit{solvable}
  if it has a scalar solution over some alphabet,

  \item \textit{scalar linearly solvable over $\Module{G}{R}$}
  if it has a scalar linear solution over the module $\Module{G}{R}$,
  and 
  \item \textit{vector linearly solvable over $\Module{G}{R}$}
  if it has a $t$-dimensional vector linear solution over the module $\Module{G}{R}$, for some $t \ge 1$.
\end{itemize}
Special cases of scalar and vector linear solvability over modules
include scalar and vector linear solvability over rings,
in which case the module is $\Module{R}{R\otimes R}$
(or equivalently, $\Module{R}{R}$, if $R$ is commutative).
The all-one's vector is an achievable rate vector for any solvable network.
We also comment that if a network has a $t$-dimensional vector solution over some alphabet $\A$,
then it has a (possibly non-linear) scalar solution over the alphabet $\A^t$,
so the network is solvable.

\subsection{Related Work}
  \label{sec:prior}

In 2000,
Ahlswede, Cai, Li, and Yeung \cite{Ahlswede-NIF}
showed that some networks
can attain higher capacities by using linear coding
at network nodes, rather than just using routing operations.
Since then,
many results on linear network coding over finite fields have been achieved.
On the other hand, the theoretical potential and limitations
of linear network coding over non-field alphabets
has been much less understood.

Li, Yeung, and Cai \cite{Li-Linear}
showed that when each of a network's 
receivers demands all of the messages
(i.e., a \textit{multicast} network),
the linear capacity over any finite field
is equal to the (nonlinear) capacity.
Ho et. al
\cite{Ho-Random} 
showed that for multicast networks, 
random fractional linear codes over finite fields
achieve the network's capacity with probability approaching one
as the block sizes increase.
Jaggi et. al
\cite{Jaggi-Algs}
developed polynomial-time algorithms
for constructing capacity-achieving fractional linear codes over finite fields
for multicast networks.
Algorithms for constructing fractional linear solutions over finite fields
for other classes of networks have also been a subject of considerable interest
(e.g., \cite{Ebrahimi-Algorithms}, \cite{Huang-Three}, \cite{Xu-Double}, and \cite{Zeng-Z}).
%

It is known (e.g., \cite{DFZ-Insufficiency}) 
that for general networks,
fractional linear codes over finite fields do not necessarily attain the network's capacity.
In fact, it was shown by
Lovett \cite{Lovett-Approximate}
that, in general, 
fractional linear network codes over finite fields 
cannot even approximate the capacity to any constant factor.
Blasiak, Kleinberg, and Lubetzky
\cite{Blasiak-Lexicopgraphic}
demonstrated a class of networks
whose capacities are
larger than their linear capacities over any finite field,
by a factor that grows polynomially with the number of messages.
Langberg and Sprintson \cite{Langberg-Hardness}
showed that, for general networks, 
constructing fractional solutions
whose rates even approximate the capacity 
to any constant factor is NP-hard.

It was shown in \cite{Cannons-Routing}
that the capacity of a network
is independent of the coding alphabet.
However, 
there are multiple examples in the literature
(e.g., \cite{CZ-NonLinear}, \cite{DFZ-Insufficiency}, \cite{DFZ-CharacteristicInequalities})
of networks whose
linear capacity over a finite field can depend on the field alphabet,
specifically by way of the characteristic of the field.
Muralidharan and Rajan \cite{Muralidharan-Polymatroids}
demonstrated that a
fractional linear solution over a finite field $\F$ exists for a network 
if and only if 
the network is associated with
a discrete polymatroid representable over $\F$.
Linear rank inequalities of vector subspaces
and linear information inequalities (e.g., \cite{Yeung-LinearInformation})
are known to be closely related
and have been shown to be useful in determining or bounding
networks' linear capacities over finite fields
(e.g., \cite{DFZ-Regions}, \cite{DFZ-CharacteristicInequalities}, and \cite{Gomez-LinearInequalities}).

Chan and Grant \cite{ChanGrant-EntropyFunctions} 
demonstrated
a duality between entropy functions and \rateregions{} of networks
and provided an alternate proof 
that fractional linear codes over finite fields 
do not necessarily attain the capacity.
The relationship between network \rateregions{}
and entropy functions has been further studied,
for example, in
\cite{ChanGrant-Capacity},
\cite{Harvey-Capacity},
\cite{Song-ZeroError},
and
\cite{Yan-Characterization}.
It has also been shown (e.g., \cite{DFZ-NonShannon})
that non-Shannon information inequalities
may be needed to determine the capacity of a network.

It was shown in \cite{ChanGrant-EntropyFunctions}
that fractional linear network codes over finite rings (and modules)
are special cases of codes generated by Abelian groups.
However, most other studies of linear capacity have generally 
been restricted to finite field alphabets.
We will consider the case where the coding alphabet is viewed, more generally, as a finite ring.

We recently showed in \cite{CZ-Commutative} and \cite{CZ-NonCommutative} that
scalar linear network codes over finite rings
can offer solvability advantages
over scalar linear network codes over finite fields
in certain cases.
Some of the results from these papers 
will be used in proofs in the present paper.

\subsection{Main Results}
  \label{sec:outline}

The remainder of the paper is outlined as follows.

In Section~\ref{sec:fractional},
we explore a connection between fractional linear codes
and vector linear codes,
which allows us to exploit
network solvability results over modules \cite{CZ-Commutative, CZ-NonCommutative}
in order to achieve capacity results over rings.
For a given network $\Network$ and rate vector $\mathbf{r}$,
we show (in Lemma~\ref{lem:kn-to-scalar}) 
there exists a network $\Network'$ that is vector linearly solvable over a given module
if and only if 
the rate vector $\mathbf{r}$ is linearly achievable for $\Network$ over the module.
In Section~\ref{ssec:dominance},
we order finite modules based on fractional solvability
and show that under certain conditions,
fractional linear solutions over a given module
imply the existence of fractional linear solutions over other modules.
The results in Sections~\ref{ssec:dominance}
and \ref{ssec:matrix-rings}
are used to show (in Lemma~\ref{lem:min-module})
that fractional linear solutions over modules
imply the existence of fractional linear solutions over modules
in which the ring of matrices over a field acts on vectors over the field.

In Section~\ref{sec:fields},
we use the results relating solvability and fractional codes 
from Section~\ref{sec:fractional}
to show our main results on \linearrateregions{} over fields.
We prove (in Theorem~\ref{thm:char-not-contained})
that for any two finite fields with different characteristics,
there exists a network whose linear rate regions over the fields
are not contained in one another.
This indicates that 
some rate vectors may only be linearly achievable over certain fields,
while other rate vectors may only be linearly achievable over other fields.
Additionally, 
for any two finite fields with different characteristics,
there exists a network whose linear capacities over the two fields are different
(Corollary~\ref{cor:char-unequal}).

We also show (in Theorem~\ref{thm:char-p})
that for any finite fields with the same characteristic,
every network's linear rate regions over the fields are equal.
In other words,
the \linearrateregion{} of any network over a field
depends only on the characteristic of the field.
Consequently, the linear capacity of any network over a field 
depends only on the characteristic of the field as well
(Corollary~\ref{cor:char-p}).
This contrasts with linear solvability over fields,
since scalar linear solvability can depend not only on the field's characteristic,
but more specifically, on the precise cardinality of the field
(e.g., see \cite[\CommLemmaFieldNoDominance]{CZ-Commutative}, \cite{Sun-BaseField}, \cite{Sun-FieldSize}).

In Section~\ref{sec:rings},
we prove our main results on linear rate regions and linear capacities
over finite rings.
We show (in Theorem~\ref{thm:ring-lin-cap}) that for any network, any finite field, 
and any finite ring whose size is divisible by the field's characteristic,
the network's \linearrateregion{} over the ring is contained within 
the network's \linearrateregion{} over the field,
and consequently the network's linear capacity over the ring is at most its linear capacity over the field
(Corollary~\ref{cor:ring-lin-capacity}).
In this sense, it suffices to restrict attention to finite fields
when choosing a coding alphabet from among all rings.
In other words, the general class of rings does not provide any benefit 
over the restricted class of finite fields,
in terms of achieving \linearrateregions{} with network coding.
In order to prove Theorem~\ref{thm:ring-lin-cap},
we show (in Theorem~\ref{thm:higher-dim-char-p})
that whenever a network has a fractional linear solution over some module
with a given rate vector,
the network has a fractional linear solution over some field
with the same rate vector
and potentially larger block sizes.

Even though Theorem~\ref{thm:ring-lin-cap} asserts non-field rings cannot provide
an increase in linear capacity over fields for \textit{all} networks,
we show (in Corollary~\ref{cor:ring-capacity-r-s})
that generally certain rings, 
smaller than a given field,
can increase the linear capacity over at least \textit{some} (but not all) networks.
In fact, we show (in Theorem~\ref{thm:ring-cap-iff})
that for any finite field and any finite ring,
there exists a network with higher linear capacity over the ring than over the field
if and only if the field's size and the ring's size are relatively prime.
Finally, we show (in Corollary~\ref{cor:module-asymptotic}) 
that whenever a network has a fractional linear solution over some ring (or module)
with a uniform rate arbitrarily close to $1$,
the network must also have a fractional linear solution over some field
with the same uniform rate.
This strengthens results in \cite{CZ-NonLinear} and \cite{DFZ-Insufficiency}
by showing that the non-linearly solvable networks presented in these papers
additionally are not asymptotically linearly solvable over rings and modules.

\clearpage

\section{Fractional and Vector Codes over Modules}
  \label{sec:fractional}

\newcommand{\butterflyFig}{
\psfrag{x1,...,xk}{\small$x_1,\ldots,x_{k_x}$}
\psfrag{y1,...,yk}{\small$y_1,\ldots,y_{k_y}$}
\psfrag{e1}{\small$e_1$}
\psfrag{en}{\small$e_{n}$}
\psfrag{n}{\small$n$}
\psfrag{e}{\small$e$}
\psfrag{x,y}{\small$x,y$}

\psfrag{0}{\small$S$}
\psfrag{1}{\footnotesize$u_1$}
\psfrag{2}{\footnotesize$u_2$}
\psfrag{3}{\footnotesize$u_3$}
\psfrag{4}{\footnotesize$u_4$}
\psfrag{5}{\scriptsize$R_1$}
\psfrag{6}{\scriptsize$R_2$}

\begin{figure}[ht]
  \begin{center}
    \leavevmode
    \hbox{\epsfxsize=0.3\textwidth\epsffile{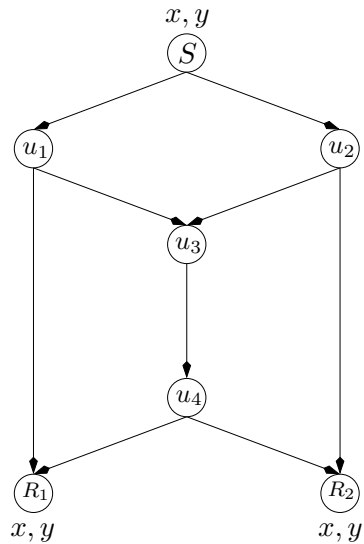}}
  \end{center}
  \caption{
    The Butterfly network has a single source node $S$,
    which generates message vectors $x$ and $y$.
    Each of the receiver nodes $R_1$ and $R_2$ demands both $x$ and $y$.
    The \linearrateregion{} of the Butterfly network is 
    $\{(r_x,r_y) \in \Q^2 \; : \; r_x,r_y \ge 0 \text{ and } r_x + r_y \le 2\}$ over any ring.
  }
\label{fig:butterfly}
\end{figure}}

\butterflyFig

Many techniques for upper bounding
network linear capacities over finite fields
(e.g., \cite{DFZ-Insufficiency,DFZ-Regions,CZ-NonLinear})
exploit linear algebra results
that sometimes do not extend to matrices over arbitrary rings.
For example, it is known (e.g., see \cite{Gupta-Transpose}) that 
the transpose of an invertible matrix over a non-commutative ring
is not necessarily invertible.%
\footnote{See \cite{Brown-Matrices}
and \cite{McDonald-LinearAlgebra}
for more information on linear algebra over rings.}
This suggests that directly computing network \linearrateregions{} and linear capacities 
over finite rings and modules
may be somewhat difficult.

One method for determining whether a network satisfies some
solvability or capacity property
is to transform the question into whether a certain
related network satisfies a corresponding property
(e.g., \cite{Kamath-Two}, \cite{Wong-Capacity}, and \cite{Wong-LinCap}).
Namely, in \cite{Wong-Capacity} and \cite{Wong-LinCap},
the authors show that determining the \rateregion{} and \linearrateregion{} of a general network
can be reduced to determining the \rateregion{} and \linearrateregion{} of a corresponding
network where each message vector is demanded by exactly one receiver (i.e., a \textit{multiple unicast} network).
In \cite{Kamath-Two},
it is shown that determining whether a multiple unicast network has a solution with a given rate vector
can be reduced to determining whether 
a corresponding unicast network with two message-receiver pairs 
has a solution with a corresponding rate vector.

We use a similar approach to relate the existence of fractional linear solutions over modules
to scalar and vector linear solvability over modules
(which was studied in \cite{CZ-Commutative} and \cite{CZ-NonCommutative}).
The results in this section
allow us to more easily relate a network's \linearrateregion{} over a ring
to the network's \linearrateregion{} over some field.

\subsection{Fractional Equivalent Network}

For any network $\Network$ with $m$ message vectors 
and integers $k_1,\dots,k_m \ge 0$ and $n \ge 1$, the following defines a new network
which is vector linearly solvable over a module $\Module{G}{R}$ if and only if
$\Network$ has a fractional linear solution over $\Module{G}{R}$ whose rate vector is $(k_1/n,\dots,k_m/n)$.
We prove this fact in Lemma~\ref{lem:kn-to-scalar}.
This network construction can be used to show many linear solvability properties
extend to the existence of fractional linear solutions.

\begin{definition}
  For any network $\Network$ with $m$ message vectors 
  and any integers $k_1,\dots,k_m \ge 0$ and $n \ge 1$,
  let $\knNetwork{k_1,\dots,k_m,n}$ denote the network $\Network$
  but with 
  \begin{itemize}\itemsep0em
    \item[(i)] each edge replaced with $n$ parallel edges, and
    \item[(ii)] the $i$th message vector replaced with $k_i$ message vectors.
  \end{itemize}
  \label{def:kn-network}
\end{definition}

\newcommand{\knButterflyFig}{
\begin{figure}[ht]
  \begin{center}
    \leavevmode
    \hbox{\epsfxsize=0.35\textwidth\epsffile{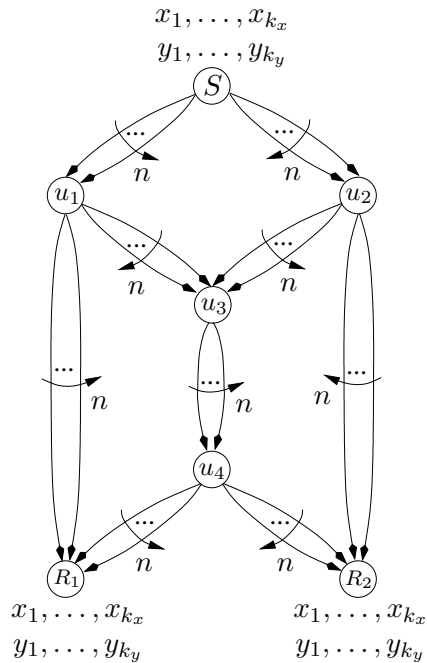}}
  \end{center}
  \caption{
    The $(k_x,k_y,n)$-Butterfly network has a single source node,
    which generates message vectors $x_1,\dots,x_{k_x}$ and $y_1,\dots,y_{k_y}$.
    Each receiver demands all of the message vectors.
    The $(k_x,k_y,n)$-Butterfly network
    is vector linearly solvable over a given ring
    if and only if $k_x+k_y \le 2n$.
  }
\label{fig:kn-butterfly}
\end{figure}}

\knButterflyFig

The \textit{Butterfly network} is defined in Figure~\ref{fig:butterfly},
and, for each $k_x,k_y \ge 0$ and $n \ge 1$, 
the \textit{$(k_x,k_y,n)$-Butterfly network} is defined in Figure~\ref{fig:kn-butterfly}.
These networks are consistent with Definition~\ref{def:kn-network},
if they are denoted by $\Network$ and $\knNetwork{k_x,k_y,n}$, respectively.

\newpage
\begin{lemma}
  Let $\Network$ be a network with $m$ message vectors,
  let $k_1,\dots,k_m \ge 0$ and $n,t \ge 1$ be integers,
  let $\Module{G}{R}$ be a module,
  and let $\knNetwork{k_1,\dots,k_m,n}$ denote
  the network in Definition~\ref{def:kn-network}
  corresponding to $\Network$ and $k_1,\dots,k_m$ and $n$.
  The network $\Network$ has a $(tk_1,\dots,tk_m,tn)$ linear solution over $\Module{G}{R}$
  if and only if 
  $\knNetwork{k_1,\dots,k_m,n}$ has a $t$-dimensional vector linear solution over $\Module{G}{R}$.
  \label{lem:kn-to-scalar}
\end{lemma}
\begin{proof}
  In a $(tk_1,\dots,tk_m,tn)$ linear code over module $\Module{G}{R}$ for network $\Network$,
  suppose a node generates
  the $l_1$th,$\dots,l_u$th message vectors and has $v$ incoming edges,
  where the $i$th message vector is an element of $G^{t k_{i}}$
  and the edge vectors are elements of $G^{t n}$.
  Then an edge function 
  $$f \, : \,
  \underbrace{G^{tk_{l_1}} \times \cdots \times G^{tk_{l_u}}}_{u \text{ message vectors}} 
  \times 
  \underbrace{G^{tn} \times \cdots \times G^{tn} }_{v \text{ edge vectors}}
  \longrightarrow G^{tn}$$
  of an outgoing edge of the node
  is of the form $f(\mathbf{x}) = A  \mathbf{x}$
  where $A$ is a $tn \times (tk_1{l_1} + \cdots + t k_{l_u} + v tn)$ matrix with entries in $R$
  and $\mathbf{x}$ is a vector over $G$ formed by concatenating the input vectors of the node.
  Let $A_1,\dots,A_n$ denote the
  $t \times (tk_1{l_1} + \cdots + t k_{l_u} + v tn)$ matrices
  such that $A$ can be written in block form as
  $$A = \left[ \begin{array}{c}
    A_1\\ 
    \vdots \\ 
    A_n
  \end{array}\right].$$
  The corresponding node in $\knNetwork{k_1,\dots,k_m,n}$
  generates $k_{l_1} + \cdots + k_{l_u}$ message vectors
  and has $v n$ incoming edge vectors.
  Define the $t$-dimensional vector code for $\knNetwork{k_1,\dots,k_m,n}$ over $\Module{G}{R}$
  by letting the edge function of the $i$th parallel corresponding outgoing edge
  be the linear mapping
  $$f_i \, : \, \underbrace{G^t \times \cdots \times G^t}_{k_{l_1}+\cdots+k_{l_u} \text{ message vectors}}
  \!\! \!\! \times \; \;
  \underbrace{G^{t} \times \cdots \times G^{t} }_{v n \text{ edge vectors}}
  \longrightarrow G^t$$
  given by $f_i(\mathbf{x}) = A_i \mathbf{x}$, where $i = 1,\dots,n$.
  The edge in the code for $\Network$ carries the same linear combination of its inputs
  as the $n$ parallel edges in the code for $\knNetwork{k_1,\dots,k_m,n}$.

  Similarly, 
  in a $(tk_1,\dots,tk_m,tn)$ code for $\Network$,
  suppose a receiver generates
  the $l_1$th,$\dots,l_u$th message vectors, has $v$ incoming edges,
  and demands $x_j$.
  Then the decoding function 
  $$d \, : \,
  \underbrace{G^{tk_{l_1}} \times \cdots \times G^{tk_{l_u}}}_{u \text{ message vectors}} 
  \times 
  \underbrace{G^{tn} \times \cdots \times G^{tn} }_{v \text{ edge vectors}}
  \longrightarrow G^{tk_j}$$
  corresponding to $x_j$
  is of the form $f(\mathbf{x}) = D  \mathbf{x}$
  where $D$ is a $tk_j \times (tk_1{l_1} + \cdots + t k_{l_u} + v tn)$ matrix
  and $\mathbf{x}$ is a vector over $G$ formed by concatenating the input vectors of the node.
  Let $D_1,\dots,D_{k_j}$ denote the
  $t \times (tk_1{l_1} + \cdots + t k_{l_u} + v tn)$ matrices
  such that $D$ can be written in block form as
  $$D = \left[ \begin{array}{c}
    D_1\\ 
    \vdots \\ 
    D_{k_j}
  \end{array}\right].$$
  The corresponding node in $\knNetwork{k_1,\dots,k_m,n}$
  generates $k_{l_1} + \cdots + k_{l_u}$ message vectors,
  has $v n$ incoming edge vectors,
  and demands the $k_j$ message vectors corresponding to $x_j$.
  Define the $t$-dimensional vector code for $\knNetwork{k_1,\dots,k_m,n}$ over $\Module{G}{R}$
  by letting the decoding function, corresponding to the $i$th such message vector,
  be the linear mapping
  $$d_i \, : \, \underbrace{G^t \times \cdots \times G^t}_{k_{l_1}+\cdots+k_{l_u} \text{ message vectors}}
  \!\! \!\! \times \; \;  
  \underbrace{G^{t} \times \cdots \times G^{t} }_{v n \text{ edge vectors}}
  \longrightarrow G^t$$
  given by $d_i(\mathbf{x}) = D_i \mathbf{x}$, where $i = 1,\dots,k_j$.
  If the function $d$ correctly reproduces its demanded message vectors in the $(tk_1,\dots,tk_m,tn)$ code for $\Network$,
  then each of $d_1,\dots,d_{k_j}$ correctly reproduces its demanded message vector in the $t$-dimensional code for 
  $\knNetwork{k_1,\dots,k_m,n}$.
  Hence, any $(tk_1,\dots,tk_m,tn)$ linear solution over a module $\Module{G}{R}$ for $\Network$
  can be translated to a $t$-dimensional vector linear solution over $\Module{G}{R}$ for $\knNetwork{k_1,\dots,k_m,n}$.
  
  A $t$-dimensional vector linear solution over the module $\Module{G}{R}$ for $\knNetwork{k_1,\dots,k_m,n}$
  can similarly be translated to a $(tk_1,\dots,tk_m,tn)$ linear solution over $\Module{G}{R}$ for $\Network$.
  In particular,
  if $f_1,\dots,f_n$ are the edge functions of the $n$ parallel edges at a node
  in a $t$-dimensional vector linear solution for $\knNetwork{k_1,\dots,k_m,n}$,
  then in the $(tk_1,\dots,tk_m,tn)$ linear code over for $\Network$, 
  define the corresponding edge function 
  to be $$f(\mathbf{x}) 
    = \left[\begin{array}{c}
      f_1(\mathbf{x}) \\
      \vdots \\
      f_n(\mathbf{x})
    \end{array} \right]. $$
  Similarly,
  if $d_1,\dots,d_{k_j}$ are the decoding functions at a node
  in a $t$-dimensional vector linear solution for $\knNetwork{k_1,\dots,k_m,n}$,
  then in the $(tk_1,\dots,tk_m,tn)$ linear code over for $\Network$, 
  define the corresponding decoding function $d(\mathbf{x})$
  to be the vector obtained by concatenating 
  $d_1(\mathbf{x}),\dots,d_{k_j}(\mathbf{x})$.
  This $(tk_1,\dots,tk_m,tn)$ linear code for $\Network$ over $\Module{G}{R}$
  is a solution,
  since the $t$-dimensional vector linear code for $\knNetwork{k_1,\dots,k_m,n}$ is a solution.
\end{proof}

When $\Module{G}{R}$
is a module and $t$ is a positive integer,
$\Module{G^t}{M_t(R)}$
denotes the module
in which the ring of all $t \times t$ matrices with entries in $R$
acts on the set of all $t$-vectors over $G$
with matrix-vector multiplication,
where multiplication of elements of $R$ with elements of $G$ is given by the action of $\Module{G}{R}$.
The following lemma shows an equivalence between fractional linear codes over modules
and fractional linear codes over these vector modules.

\begin{lemma}
  Let $\Module{G}{R}$ be a module,
  let $\Network$ be a network,
  and let $k_1,\dots,k_m \ge 0$ and $n,t \ge 1$ be integers.
  Network $\Network$ has a $(k_1,\dots,k_m,n)$ linear solution over
  $\Module{G^t}{M_t(R)}$
  if and only if
  $\Network$ has a $(t k_1,\dots,t k_m, tn)$ linear solution over $\Module{G}{R}$.
  \label{lem:vector_module}
\end{lemma}
\begin{proof}
  This lemma follows from the fact that a scalar linear solution over $\Module{G^t}{M_t(R)}$
  is equivalent to a $t$-dimensional vector linear solution over $\Module{G}{R}$.
  In particular,
  in both a scalar linear code over $\Module{G^t}{M_t(R)}$
  and a $t$-dimensional vector linear code over $\Module{G}{R}$,
  inputs to a node are $t$-vectors over $G$
  and outputs carry linear combinations of the inputs,
  where the coefficients that describe the linear combination are $t \times t$ 
  matrices over $R$.
  Any scalar linear solution over $\Module{G^t}{M_t(R)}$
  can be translated to a $t$-dimensional vector linear solution over $\Module{G}{R}$
  and vice versa.
  This idea generalizes to fractional linear solutions:
  \begin{align*}
    &\text{$\Network$ has a $(k_1,\dots,k_m,n)$ linear solution over $\Module{G^t}{M_t(R)}$} 
      \\
    & \Longleftrightarrow \;
    \text{$\knNetwork{k_1,\dots,k_m,n}$ has a scalar linear solution over $\Module{G^t}{M_t(R)}$}
     & & \Comment{Lemma~\ref{lem:kn-to-scalar}}
      \\
    & \Longleftrightarrow \;
    \text{$\knNetwork{k_1,\dots,k_m,n}$ has a $t$-dimensional linear solution over $\Module{G}{R}$}
      \\
    & \Longleftrightarrow \;
    \text{$\Network$ has a $(tk_1,\dots,tk_m,tn)$ linear solution over $\Module{G}{R}$}
     & & \Comment{Lemma~\ref{lem:kn-to-scalar}}.
  \end{align*}

\end{proof}

\subsection{Fractional Dominance} \label{ssec:dominance}

\begin{definition}
  Let $\Module{G}{R}$ and $\Module{H}{S}$ be modules.
  We say that
  \begin{itemize}
  \item[(a)]
  $\Module{H}{S}$ \textit{\scaDom{}} $\Module{G}{R}$
  if every network 
  with a scalar linear solution over $\Module{G}{R}$
  also has a scalar linear solution over $\Module{H}{S}$,

  \item[(b)]
  $\Module{H}{S}$ \textit{\fracDom{}} $\Module{G}{R}$
  if for each $k_1,\dots,k_m \ge 0$ and $n \ge 1$, 
  every network 
  with a $(k_1,\dots,k_m,n)$ linear solution over $\Module{G}{R}$
  also has a $(k_1,\dots,k_m,n)$ linear solution over $\Module{H}{S}$.
  \end{itemize}
  \label{def:Dominance}
\end{definition}

\begin{remark}
  Any left-sided fractional linear code over a ring
  can be viewed as a two-sided fractional linear code over the ring
  in which the inputs are multiplied on the right by the identity element,
  so
  the module $\Module{R}{R \otimes R^{op}}$
  \fracDom{} $\Module{R}{R}$
  for every finite ring $R$.
  Furthermore, if $R$ is commutative,
  then any two-sided fractional linear code over $R$ can equivalently be written
  as a left-sided fractional linear code over $R$,
  which implies $\Module{R}{R}$ \fracDom{} $\Module{R}{R \otimes R^{op}}$.
\end{remark}

We also comment that if $R$ and $S$ are finite rings
such that
$\Module{S}{S \otimes S^{op}}$ \fracDom{}
$\Module{R}{R \otimes R^{op}}$,
then for each network $\Network$, we have
$$\LinRegion{\Network}{S}
\supseteq
\LinRegion{\Network}{R}
\; \; 
\text{ and }
\; \;
\UniLinCap{\Network}{S}
\ge
\UniLinCap{\Network}{R}.$$

The following lemma shows that \scalardominance{}
and \fractionaldominance{}
of modules are, in fact, equivalent.
However, it is cleaner to prove results on \scalardominance,
as the block sizes of the message vectors and edge vectors
are all one,
and we can use results from \cite{CZ-NonCommutative}.

\begin{lemma}
  Let $\Module{G}{R}$ and $\Module{H}{S}$ be modules.
  $\Module{H}{S}$ \scaDom{} $\Module{G}{R}$ if and only if 
  $\Module{H}{S}$ \fracDom{} $\Module{G}{R}$.
  \label{lem:scalar-dominance}
\end{lemma}
\begin{proof}
  It follows immediately from the definition that
  $\Module{H}{S}$ \fracDom{} $\Module{G}{R}$
  implies $\Module{H}{S}$ \scaDom{} $\Module{G}{R}$.
  To prove the converse,
  suppose $\Module{H}{S}$ \scaDom{} $\Module{G}{R}$.
  Let $\Network$ be a network with $m$ message vectors, 
  let $k_1,\dots,k_m \ge 0$ and $n \ge 1$ be integers,
  and let $\knNetwork{k_1,\dots,k_m,n}$ be the network in Definition~\ref{def:kn-network}
  corresponding to $\Network$, $k_1,\dots,k_m$, and $n$.
  Then
  \begin{align*}
    & \text{$\Network$ has a $(k_1,\dots,k_m,n)$ linear solution over $\Module{G}{R}$} \\
    & \implies \;
     \text{$\knNetwork{k_1,\dots,k_m,n}$ has a scalar linear solution over $\Module{G}{R}$}
      & & \Comment{Lemma~\ref{lem:kn-to-scalar}}
      \\
    & \implies \;
     \text{$\knNetwork{k_1,\dots,k_m,n}$ has a scalar linear solution over $\Module{H}{S}$}
      & & \Comment{$S$ \scaDom{} $R$}
      \\
    & \implies \;
       \text{$\Network$ has a $(k_1,\dots,k_m,n)$ linear solution over $\Module{H}{S}$}
      & & \Comment{Lemma~\ref{lem:kn-to-scalar}}.
  \end{align*}
  Hence, for any network, 
  any fractional linear solution over $\Module{G}{R}$
  implies the existence of a fractional linear solution over $\Module{H}{S}$
  with the same block sizes.
\end{proof}

\begin{definition}
An $R$-module $G$ is \textit{faithful} if 
for each $r \in R\backslash\{0\}$,
there exists $g \in G$
such that $r \act g \ne 0$.
\label{def:faithful}
\end{definition}

Lemmas~\ref{lem:same_ring}, \ref{lem:faithful_module}, and \ref{lem:ModHomomorphism}
follow immedately from Lemma~\ref{lem:scalar-dominance} and results from \cite{CZ-NonCommutative},
and we include their proofs in the appendix for reference.
Lemma~\ref{lem:same_ring}
shows that, for a fixed ring $R$, 
fractional linear solutions over faithful $R$-modules
induce fractional linear solutions over every other $R$-module.
Lemma~\ref{lem:faithful_module} shows that fractional linear solutions over non-faithful modules
induce fractional linear solutions over some faithful module.
Lemma~\ref{lem:ModHomomorphism} shows that ring homomorphisms
also induce fractional dominance.

\begin{lemma}
  Let $R$ be a fixed ring,
  let $G$ be a faithful $R$-module,
  and let $H$ be an $R$-module.
  Then $\Module{H}{R}$ \fracDom{} $\Module{G}{R}$.
  \label{lem:same_ring}
\end{lemma}
\newcommand{\LemSameRingProof}{
\begin{proof}[Proof of Lemma~\ref{lem:same_ring}]
  Let $\Network$ be a network that is scalar linearly solvable over 
  the faithful $R$-module $(G,\oplus)$.
  Any scalar linear solution for $\Network$ over the $R$-module $(G,\oplus)$
  is a scalar linear solution for $\Network$ over any other $R$-module.

  To see this,
  let $z_1,\dots,z_m \in G$ denote the messages of $\Network$,
  and suppose a node in $\Network$ has inputs 
  $x_1,\dots,x_n \in G$ in a scalar linear solution over $\Module{G}{R}$, 
  where, for each $i = 1,\dots,n$,
  $$x_i = \bigoplus_{j=1}^{m} (B_{i,j} \act z_j)$$
  for some $B_{i,1},\dots,B_{i,m} \in R$.
  Then for each output $y \in G$ of this node,
  there exist constants $C_1,\dots,C_n \in R$ such that
  \begin{align*}
    y & = \bigoplus_{i=1}^{n} (C_{i} \act x_i) \\
      & =\bigoplus_{i=1}^{n} \bigoplus_{j=1}^{m} ((C_{i} B_{i,j}) \act z_j)\\
      & = \bigoplus_{j=1}^{m} \left(\left( \sum_{i = 1}^{n} C_i B_{i,j} \right) \act z_j \right).
  \end{align*}

  Now let $H$ be any $R$-module with action $\newaction$, and
  suppose the corresponding inputs to the node in the scalar linear code over $\Module{H}{R}$
  are $x_1',\dots,x_n' \in H$ and can be written in terms of the messages
  $z_1',\dots,z_m' \in H$ in the following way
  $$x_i' = \bigoplus_{j=1}^{m} (B_{i,j} \newaction z_j').$$
  Then the corresponding output $y' \in R$ of the node is of the form
  \begin{align*}
    y' & = \bigoplus_{i=1}^{n} (C_{i} \newaction x_i')\\
      & = \bigoplus_{i=1}^{n} \bigoplus_{j=1}^{m} ((C_{i} B_{i,j}) \newaction z_j') \\
      & = \bigoplus_{j=1}^{m} \left(\left( \sum_{i = 1}^{n} C_i B_{i,j} \right) \newaction z_j' \right)
  \end{align*}
  so by induction, every edge and decoding function in the scalar linear code over $\Module{H}{R}$
  is the same linear combination of the messages
  as in the scalar linear solution over $\Module{G}{R}$.
  
  $G$ is a faithful $R$-module, 
  so $1$ and $0$ are the only elements of $R$ such that $1 \act g = g$ and $0 \act g = 0$ for all $g \in G$.
  Hence it must be the case that decoding functions in the scalar linear solution over $\Module{G}{R}$ 
  are of the form
  $$(1 \act z_i) \oplus \bigoplus_{\substack{j = 1 \\ j \ne i}}^n (0 \act z_j) = z_i$$
  so it must be the case
  that the corresponding decoding function in scalar the linear code over $\Module{H}{R}$
  is
  $$(1 \newaction z_i') \oplus \bigoplus_{\substack{j = 1 \\ j \ne i}}^n (0 \newaction z_j') = z_i'.$$
  Hence, each receiver can linearly recover its demands,
  so the scalar linear code over $\Module{H}{R}$ is, in fact, a solution.
  This implies that $\Module{H}{R}$ \scaDom{} $\Module{G}{R}$,
  which along with Lemma~\ref{lem:scalar-dominance},
  shows that $\Module{H}{R}$ \fracDom{} $\Module{G}{R}$.
\end{proof}}

In \cite{CZ-NonCommutative},
an example was given in which
a network has a scalar linear solution over a non-faithful $R$-module
but does not have any scalar linear solutions over another $R$-module.
This shows the importance of the faithfulness of the module in Lemma~\ref{lem:same_ring}.

\begin{lemma}
  Let $G$ be an $R$-module.
  There exists a finite ring $S$ such that
  $G$ is a faithful $S$-module,
  and $\Module{G}{S}$ \fracDom{} $\Module{G}{R}$.
  \label{lem:faithful_module}
\end{lemma}
\newcommand{\lemFaithfulModuleProof}{
\begin{proof}[Proof of Lemma~\ref{lem:faithful_module}]
  We use ideas from \cite[p. 2750]{DFZ-Insufficiency} here.
  Let $$J = \{r \in R \; : \; r \act g = 0, \; \forall g \in G\}$$
  which is easily verified to be a two-sided ideal of $R$.
  Let $S = R/J$.
  It can also be verified that
  $G$ is an $S$-module with action $\newaction:S \times G \to G$ given by
  $$(r + J) \newaction g = r \act g.$$
  If $(r + J), (s + J) \in S$ are such that
  $$(r + J) \newaction g = (s + J) \newaction g$$
  for all $g \in G$,
  then $(r-s) \act g = 0$,
  which implies $(r - s) \in J$.
  Hence $(r + J) = (s + J)$,
  so the ring $S$ acts faithfully on $G$.
  A faithful module requires different elements of the ring
  to yield different functions when acting on elements of the group.
  Since $G$ is finite,
  the number of such functions must be finite,
  which implies the ring $S$ must also be finite.

  Suppose a network $\Network$ is scalar linearly solvable over $\Module{G}{R}$.
  Every output $y'$ in the solution over $\Module{G}{R}$ is of the form
  \begin{align}
  y' = (C_1 \act x_1) \oplus \cdots \oplus (C_m \act x_m)
  \label{eq:t}
  \end{align}
  where the $x_i$'s are the parent node's inputs and the $C_i$'s are constants from $R$.
  Form a linear code over $\Module{G}{S}$ replacing each coefficient $C_i$ in \eqref{eq:t}
  by $(C_i + J)$.
  Let $y$ be the edge symbol in the code over $\Module{G}{S}$
  corresponding to $y'$ in the code over $\Module{G}{R}$.
  Then
  \begin{align*}
  y &= ((C_1 + J) \newaction x_1) \oplus \cdots \oplus ((C_m + J) \newaction x_m)  \\
    &= (C_1 \act x_1) \oplus \cdots \oplus (C_m \act x_m) = y'.
  \end{align*}
  Thus, whenever an edge function in the solution over $\Module{G}{R}$ outputs the symbol $y'$,
  the corresponding edge function in the code over $\Module{G}{S}$ will output the same symbol $y'$.
  Likewise, whenever $x$ is an input to an edge function in the solution over $\Module{G}{R}$,
  the corresponding input of the corresponding edge function in the code over $\Module{G}{S}$ 
  will be the same symbol $x$.
  The same argument holds for the decoding functions in the code over $\Module{G}{S}$, so each
  receiver will correctly obtain its corresponding demands in the code over $\Module{G}{S}$.
  Hence, the code over $\Module{G}{S}$ is a linear solution for $\Network$.

  This implies $\Module{G}{S}$ \scaDom{} $\Module{G}{R}$,
  which along with Lemma~\ref{lem:scalar-dominance},
  implies $\Module{G}{S}$ \fracDom{} $\Module{G}{R}$
\end{proof}}

A \textit{ring homomorphism} is a mapping $\phi$
from a ring $R$ to a ring $S$ such that for all $a,b \in R$
\begin{align*}
  \phi(a + b) &= \phi(a) + \phi(b) \\
  \phi(ab) &= \phi(a) \phi(b) \\
  \phi(1_R) &= 1_S
\end{align*}
where $1_R$ and $1_S$ are the multiplicative identities of $R$ and $S$, respectively.
It follows from this definition that $\phi(0_R) = 0_S$,
where $0_R$ and $0_S$ are the additive identities of $R$ and $S$, respectively.

\begin{lemma}
  Let $\phi: R \to S$ be a ring homomorphism,
  let $G$ be a faithful $R$-module,
  and let $H$ be an $S$-module.
  Then $\Module{H}{S}$ \fracDom{} $\Module{G}{R}$.
  \label{lem:ModHomomorphism}
\end{lemma}
\newcommand{\ProofOfLemModHomomorphism}{
\begin{proof}[Proof of Lemma~\ref{lem:ModHomomorphism}]
  Let $H$ be an $S$-module and
  define a mapping 
  $$\newaction: R \times H \to H$$ by
  $r \newaction h = \phi(r) \act h$,
  where $\act$ is the action of $\Module{H}{S}$.
  One can verify 
  that $H$ is an $R$-module under $\newaction$.
  Now, let $G$ be a faithful $R$-module,
  and suppose
  $\Network$ has a linear solution over $\Module{G}{R}$.
  By Lemma~\ref{lem:same_ring},
  $\Network$ is scalar linearly solvable over $\Module{H}{R}$,
  so every output $y' \in H$ in the solution over $\Module{H}{R}$ is of the form
  \begin{align}
    y' &= (C_1 \newaction x_1) \oplus \cdots \oplus (C_m \newaction x_m)
  \label{eq:100}
  \end{align}
  where $x_1,\dots,x_m \in H$ are the parent node's inputs 
  and $C_1,\dots,C_m \in R$ are constants.

  Form a linear code for $\Network$ over $\Module{H}{S}$ 
  by replacing each coefficient $C_i$ in \eqref{eq:100} by $\phi(C_i)$.
  Let $y \in H$ be the output in the code over $\Module{H}{S}$
  corresponding to $y'$ in the code over $\Module{H}{R}$.
  Then
  \begin{align*}
  y &= (\phi(C_1) \act x_1) \oplus \cdots \oplus (\phi(C_m) \act x_m)  \\
    &= (C_1 \newaction x_1) \oplus \cdots \oplus (C_m \newaction x_m) = y'.
  \end{align*}
  By induction, whenever an edge function in the solution over $\Module{H}{R}$ outputs the symbol $y'$,
  the corresponding edge function in the code over $\Module{H}{S}$ will output the same symbol $y'$.
  Likewise, whenever $x$ is an input to an edge function in the solution over $\Module{H}{R}$,
  the corresponding input of the corresponding edge function in the code over $\Module{H}{S}$ 
  will be the same symbol $x$.
  The same argument holds for the decoding functions in the code over $\Module{H}{S}$, so each
  receiver will correctly obtain its corresponding demands in the code over $\Module{H}{S}$.
  Hence, the code over $\Module{H}{S}$ is a linear solution for $\Network$.

  This implies that $\Module{H}{R}$ \scaDom{} $\Module{G}{R}$,
  which along with Lemma~\ref{lem:scalar-dominance},
  shows that $\Module{H}{R}$ \fracDom{} $\Module{G}{R}$.
\end{proof}}

By the fundamental theorem of finite Abelian groups,
every finite Abelian group
is isomorphic to a direct product of cyclic groups
whose sizes are prime powers (with component-wise addition)
\cite[p. 161]{Dummit-Algebra}.
As an example,
$\Z_{12} \cong \Z_4 \times \Z_3$.
The following lemma shows that if a finite Abelian group can be written
as a direct product of Abelian groups $G$ and $H$ whose sizes are relatively prime,
then whenever $G \times H$ is an $R$-module for some ring $R$,
the ring $R$ acts on $G \times H$ component-wise.
This implies that $G$ and $H$ are also $R$-modules.
Since fractional linear solutions over faithful $R$-modules
induce fractional linear solutions over every other $R$-module,
this is a useful tool for showing fractional dominance.

\begin{lemma}
  Let $G$ and $H$ be finite groups such that $|G|$ and $|H|$ are relatively prime,
  and let $G \times H$ be some $R$-module.
  Then $G$ and $H$ are also $R$-modules.
  \label{lem:submodules-1}
\end{lemma}
\begin{proof}
  Let $g \in G$ and $r \in R$, and suppose $r \act (g,0) = (g_r,h_r) \in G \times H$.
  It follows from Lagrange's theorem of finite groups (e.g., \cite[p. 45]{Dummit-Algebra})
  that $|G| g = \underbrace{g \oplus \cdots \oplus g}_{|G| \text{ times}} = 0$,
  so
  \begin{align*}
    (0,0) &= r \act (0,0) 
       = r \act ( |G| \, g, 0)
       = |G| \, r \act (g, 0) = |G| \, (g_r,h_r) = (|G| g_r, |G| h_r) 
       = (0, |G| h_r).
  \end{align*}
  Since $|G|$ and $|H|$ are relatively prime, 
  it follows from Cauchy's theorem of finite groups (e.g., \cite[p. 93]{Dummit-Algebra})
  that $H$ contains no non-identity elements whose order divides $|G|$,
  so it must be the case that $h_r = 0$.
  Similarly,
  for each $h \in H$ and each $r \in R$,
  there exists $h_r \in H$ such that
  $r \act (0,h) = (0,h_r)$.
  This implies $R$ acts on $G \times H$ component-wise.
  In other words,
  if $r \act (g,h) = (g_r, h_r)$,
  then $r \act (g,0) = (g_r,0)$
  and $r \act (0,h) = (0,h_r)$.
  Thus the mapping $\newaction: R \times G \to G$
  given by $r \newaction g = g_r$ satisfies the properties of an action,
  so $G$ is an $R$-module with action $\newaction$.
  It can similarly be shown $H$ is an $R$-module.
\end{proof}

We comment that Lemma~\ref{lem:submodules-1}
does not extend to finite groups whose sizes are not relatively prime.
As an example,
the field $\GF{4}$ acts on its own additive group $(\GF{4},+)$
by multiplication in the field.
If the elements of $\GF{4}$ are represented as
$\{0,1,\alpha,\alpha+1\}$
where $\alpha^2 = \alpha + 1$,
then for all $(a_0 + \alpha a_1), (b_0 + \alpha b_1) \in \GF{4}$
$$(a_0 + \alpha a_1) \, (b_0 + \alpha b_1) 
  = a_0 b_0 + a_1 b_1 + \alpha (a_0 b_1 + a_1 b_0 + a_1 b_1).$$
The additive group of $\GF{4}$
is isomorphic to the set $\GF{2} \times \GF{2}$ with component-wise addition in $\GF{2}$,
so $\GF{4}$ acts on $\GF{2} \times \GF{2}$ by
$$(a_0 + \alpha \, a_1) \cdot (b_0,b_1)
  = (a_0 b_0 + a_1 b_1, \, a_0 b_1 + a_1 b_0 + a_1 b_1).$$
This action is not component-wise, since $(1 + \alpha) \act (1,0) = (1,1)$
and $\alpha \act (0,1) = (1,1)$.

If $\GF{4}$ acts on $\GF{2}$,
then the action must be such that
$1 \act a = a$ and $0 \act a = 0$ for all $a \in \GF{2}$
and $x \act 0 = 0$ for all $x \in \GF{4}$.
If $\alpha \act 1 = 1$, then 
$$0 = 1 + 1 = (\alpha \act 1) + (1 \act 1) = (\alpha + 1) \act 1
  = (\alpha^2) \act 1 = \alpha \act ( \alpha \act 1) = \alpha \act 1 = 1$$
  which is a contradiction.
If $\alpha \act 1 = 0$, then
$$1 = 0 + 1 = (\alpha \act 1) + (1 \act 1) = (\alpha + 1) \act 1
  = (\alpha^2) \act 1 = \alpha \act (\alpha \act 1) = \alpha \act 0 = 0$$
  which is a contradiction.
Thus $\GF{2}$ cannot be a $\GF{4}$-module,
but as shown above, $\GF{2} \times \GF{2}$ is a $\GF{4}$-module.

\newpage
\subsection{Matrix Rings over Fields} \label{ssec:matrix-rings}

If a ring $R$ has a proper two-sided ideal $I$,
then there is a surjective homomorphism from
$R$ to $R/I$.
It is known (e.g., \cite[p. 20]{McDonald-FiniteRings}) 
that every finite ring with no proper two-sided ideals
is isomorphic to some ring of matrices over a finite field.
In fact, every finite ring $R$
has a two-sided ideal $I$
such that $R/I$ is a matrix ring over a field.
This implies the following lemma,
which was more formally shown in \cite{CZ-NonCommutative}.

\begin{lemma}{\cite[\NonCommSimpleHomoLemmas]{CZ-NonCommutative}:}
  Let $R$ be a finite ring. 
  There exists a positive integer $t$,
  a finite field $\F$,
  and a surjective homomorphism from $R$ to $M_t(\F)$.
  \label{lem:simple}
\end{lemma}

Lemmas~\ref{lem:ModHomomorphism} and \ref{lem:simple}
together imply that fractional linear solutions over modules
induce fractional linear solutions over modules
in which the ring is a matrix ring over a field.
The following lemma proves a result on the cardinality 
of such modules.

\begin{lemma}
  Let $\F$ be a finite field and $t$ a positive integer.
  If $G$ is a finite non-zero $M_t(\F)$-module,
  then $|\F|^t$ divides $|G|$.
  \label{lem:smaller_module}
\end{lemma}
\begin{proof}
  Since $G$ is finite and non-zero,
  $G$ contains a submodule with no proper submodules
  (possibly $G$ itself).
  It is known (e.g., \cite[Theorem 3.3 (2), p. 31]{Lam-Noncommutative}) that 
  $\F^t$ is the only $M_t(\F)$-module with no proper submodules,
  so $\F^t$ is a submodule of $G$.
  Hence by Lagrange's theorem of finite groups,
  $|\F|^t$ divides $|G|$.
\end{proof}

Lemma~\ref{lem:min-module}
shows that every module
is \fracdominated{} by a module
whose group is the set of $t$ vectors over some field
and whose ring is the set of all $t \times t$ matrices over the field.
In network coding, arbitrarily large block sizes may be needed to achieve a solution with a particular rate.
Das and Rai \cite{Das-Fractional} showed that
for each $k,n \ge 1$ and each $t \ge 2$, 
there exists a network that has a $(tk,\dots,tk, tn)$ linear solution over any finite field,
yet the network has no $(sk,\dots,sk,sn)$ linear solution over any finite field when $s < t$.
It was also shown in \cite{CZ-NonCommutative}
that for each $t \ge 2$, 
there exist networks with scalar linear solutions over certain rings
but with no $s$-dimensional vector linear solutions over any field whenever $s < t$.
This suggests that the quantity $t$ in Lemma~\ref{lem:min-module}
may need to be arbitrarily large.

\begin{lemma}
  Let $\Module{G}{R}$ be a module.
  For each prime $p$ that divides $|G|$,
  there exists a finite field $\F$ of characteristic $p$ and a positive integer $t$
  such that 
  $\Module{\F^t}{M_t(\F)}$ \fracDom{} $\Module{G}{R}$.
  \label{lem:min-module}
\end{lemma}
\begin{proof}
  By Lemma~\ref{lem:faithful_module}
  there exists a finite ring $S$ such that
  the faithful module $\Module{G}{S}$ \fracDom{} $\Module{G}{R}$.
  By the fundamental theorem of finite Abelian groups,
  the group $G$ is isomorphic to a direct product
  of Abelian groups whose sizes are prime powers,
  and since $\Div{p}{|G|}$,
  the size of at least one of these groups is a power of $p$.
  Let $H$ be the direct product of all such groups whose sizes are powers of $p$.
  Then there exists a finite group $G'$ such that
  $G \cong G' \times H$ and $|G'|$ and $|H|$ are relatively prime.
  Hence by Lemma~\ref{lem:submodules-1},
  $H$ is also an $S$-module,
  and since $G$ is a faithful $S$-module,
  by Lemma~\ref{lem:same_ring},
  the module $\Module{H}{S}$ \fracDom{} $\Module{G}{S}$.
  
  By Lemma~\ref{lem:faithful_module},
  there exists a finite ring $S'$ such that $H$ is a faithful $S'$-module
  and $\Module{H}{S'}$ \fracDom{} $\Module{H}{S}$.
  By Lemma~\ref{lem:simple},
  there exists a positive integer $t$,
  a finite field $\F$,
  and a surjective homomorphism from $S'$ to $M_t(\F)$.
  By Lemma~\ref{lem:ModHomomorphism},
  the module $\Module{H}{S'}$
  is \fracdominated{} by every $M_t(\F)$-module,
  and the ring $M_t(\F)$ acts on the of all $t$-vectors over $\F$
  by matrix-vector multiplication over $\F$,
  so $\Module{\F^t}{M_t(\F)}$ \fracDom{} $\Module{H}{S'}$.
  The proof of Lemma~\ref{lem:ModHomomorphism}
  also implies $H$ is an $M_t(\F)$-module,
  so Lemma~\ref{lem:smaller_module}
  implies $\Div{|\F|^t}{|H|}$.
  Since $|H|$ is a power of $p$,
  this implies $\F$ is a field of characteristic $p$.
  Finally, by the transitivity of \fractionaldominance,
  $\Module{\F^t}{M_t(\F)}$ \fracDom{} $\Module{G}{R}$.
\end{proof}

Lemma~\ref{lem:char_p} uses ideas similar to those in \cite[Proposition 1]{Sun-VL}
and \cite{Ebrahimi-Algorithms},
and we include a proof for completeness.
This lemma, along with Lemma~\ref{lem:vector_module},
implies that a fractional linear solution over any non-prime finite field
induces a fractional linear solution over the corresponding prime field
with the same rate vector.
A fractional linear solution over a field $\F$ is equivalent to 
a fractional linear solution over the faithful module
$\Module{\F}{\F}$,
since $\GF{\F}$ is commutative.

\begin{lemma}
  Let $q$ be a prime power and $t$ a positive integer.
  Then $\Module{\GF{q}^t}{M_t(\GF{q})}$ \fracDom{} $\Module{\GF{q^t}}{\GF{q^t}}$.
  \label{lem:char_p}
\end{lemma}
\begin{proof}
  It is known (e.g., see \cite[p. 531]{Dummit-Algebra})
  that every extension field $\GF{q^t}$
  is isomorphic to a set of $t \times t$ matrices over $\GF{q}$.
  This implies there exists an injective homomorphism from $\GF{q^t}$ to $M_t(\GF{q})$.
  By Lemma~\ref{lem:ModHomomorphism},
  any network with a fractional linear solution over $\Module{\GF{q^t}}{\GF{q^t}}$
  also has a fractional linear solution over any $M_t(\GF{q})$-module.
  In particular,
  $\Module{\GF{q}^t}{M_t(\GF{q})}$ \fracDom{} $\Module{\GF{q^t}}{\GF{q^t}}$.
\end{proof}

\clearpage

\section{Linear Rate Regions over Fields} \label{sec:fields}

\psfrag{e0}{$e_0$}
\psfrag{e1}{$e_1$}
\psfrag{en}{\!\!\footnotesize$e_{m+1}$}
\psfrag{ex}{$e$}
\psfrag{y0}{$x_0$}
\psfrag{y1}{$x_1$}
\psfrag{yn}{$x_{m+1}$}
\psfrag{u0}{\small$u_0$}
\psfrag{u1}{\small$u_1$}
\psfrag{un}{\!\!\footnotesize$u_{m+1}$}
\psfrag{ux}{\small$u$}
\psfrag{v0}{\small$v_0$}
\psfrag{v1}{\small$v_1$}
\psfrag{vn}{\!\!\footnotesize$v_{m+1}$}
\psfrag{vx}{\small$v$}
\psfrag{R0}{\small$R_0$}
\psfrag{R1}{\small$R_1$}
\psfrag{Rn}{\!\scriptsize$R_{m+1}$}
\psfrag{Rx}{\small\,\,$R$}
\psfrag{S0}{\small$S_0$}
\psfrag{S1}{\small$S_1$}
\psfrag{Sn}{\!\!\scriptsize$S_{m+1}$}

\begin{figure}[ht]
  \begin{center}
    \leavevmode
    \hbox{\epsfxsize=0.5\textwidth\epsffile{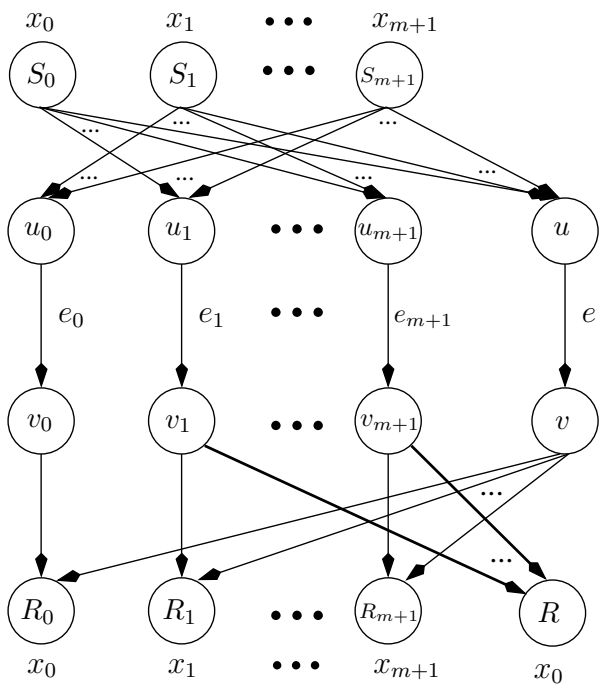}}
  \end{center}
  \caption{
    The \charNetwork{$m$} has source nodes $S_0,S_1,\dots,S_{m+1}$
    which generate message vectors $x_0,x_1,\dots,x_{m+1}$, respectively.
    Node $u$ has a single incoming edge
    from each source node,
    and the edge connecting nodes $u$ and $v$ carries the edge vector $e$.
    For each $i = 0,1,\dots,m+1$,
    node $u_i$ has a single incoming edge
    from each source node, except $S_i$.
    The edge connecting nodes $u_i$ and $v_i$ carries edge vector $e_i$.
    The receiver $R_i$ demands $x_i$ and has an incoming edge from node
    $v_i$ and an incoming edge from $v$.
    The receiver $R$ demands $x_0$
    and has an incoming edge from each of nodes $v_1,\dots,v_{m+1}$.
  }
\label{fig:p-network}
\end{figure}

We define, for each integer $m \ge 2$,  the \textit{\charNetwork{$m$}}  in Figure~\ref{fig:p-network}.
The \charNetwork{$m$} is denoted by $\Network_2(m,1)$ in \cite{CZ-NonLinear},
with a slight relabeling of sources,
and the \charNetwork{$m$} is known to be vector linearly solvable over a field
if and only if the characteristic of the field divides $m$.
When $m = 2$,
this network exhibits solvability properties similar to those of the Fano network \cite{DFZ-NonShannon}.

Let $R$ be a finite ring whose characteristic divides $m$.
Then $m = 0$ in $R$,
and the following scalar linear code:
\begin{align*}
  e &= \sum_{j=0}^{m+1} x_j 
  \; \; \text{ and } \; \; 
  e_i = \sum_{\substack{j=0 \\ j \ne i}}^{m+1} x_j
\end{align*}
over $R$ is a solution for the \charNetwork{$m$},
where $i = 0,1,\dots,m+1$, and the receivers linearly recover their demands as follows
\begin{align*}
  R_i : \; \ e - e_i &= x_i \\
  R : \; \ \ \sum_{i=1}^{m+1} e_i 
      & =  x_0 + m \sum_{i=0}^{m+1} x_i  \\
      &  = x_0 
        & & \Comment{$\Div{\Char{R}}{m}$}.
\end{align*}
This code relies on the fact $m = 0$ in $R$,
and it turns out the \charNetwork{$m$} has no scalar linear solutions
over any ring whose characteristic does not divide $m$
(see \cite[Lemma IV.6]{CZ-NonLinear}).

\begin{lemma}{\cite[Lemma IV.7]{CZ-NonLinear}:}
  For each $m \ge 2$ and each finite field $\F$,
  the linear capacity of the \charNetwork{$m$} is
  \begin{itemize}
    \item equal to $1$, whenever $\Div{\Char{\F}}{m}$, and
    \item upper bounded by $1-\frac{1}{2m+3}$,
    whenever $\NDiv{\Char{\F}}{m}$.
  \end{itemize}
  \label{lem:char-p-network}
\end{lemma}

\subsection{Comparing Linear Rate Regions over Different Fields}

It follows from Lemma~\ref{lem:char-p-network}
that certain fields may yield strictly larger
linear capacities for some networks
than other fields.
In particular,
whenever the characteristics of two finite fields are different,
there exists some network whose linear capacities over the fields differ.

\begin{corollary}
  If $\F$ and $\K$ are finite fields with different characteristics,
  then there exist networks $\Network_1$ and $\Network_2$,
  such that 
  $\UniLinCap{\Network_1}{\F} > \UniLinCap{\Network_1}{\K}$
  and
  $\UniLinCap{\Network_2}{\K} > \UniLinCap{\Network_2}{\F}.$
 \label{cor:char-unequal}
\end{corollary}
\begin{proof}
  Suppose $\Char{\F} = p \ne q = \Char{\K}$
  and let $\Network_1$ and $\Network_2$ be the 
  \charNetwork{$p$} and the \charNetwork{$q$}, respectively.
  Then by Lemma~\ref{lem:char-p-network},
  $\UniLinCap{\Network_1}{\F} = 1$
  and 
  $\UniLinCap{\Network_1}{\K} \le 1 - \frac{1}{2p+3}$.
  Similarly,
  $\UniLinCap{\Network_2}{\K} = 1$
  and 
  $\UniLinCap{\Network_2}{\F} \le 1 - \frac{1}{2q+3}$.
\end{proof}

In \cite{DFZ-Regions},
it was shown that for any finite fields $\F$ and $\K$ of even and odd characteristic, respectively:
(i) the linear rate region of the non-Fano network over $\F$
is a proper subset of its linear rate region over $\K$,
and (ii) the linear rate region of the Fano network over $\K$
is a proper subset of its linear rate region over $\F$.
In these instances,
it is strictly ``better'' to use an even/odd characteristic field
instead of an odd/even characteristic field.
However,
the following theorem demonstrates that it may not always be the case
that one field is necessarily ``better'' than the other
for a particular network.
In particular,
for some networks,
some rate vectors may only be linearly achievable over certain fields
while other rate vectors may only be linearly achievable over other fields.

\begin{theorem}
For any two finite fields with different characteristics,
there exists a network whose linear rate regions over the fields
do not contain one another.
 \label{thm:char-not-contained}
\end{theorem}
\begin{proof}
  A \textit{disjoint union} of networks refers to a new network 
  whose nodes/edges/sources/receivers
  are the disjoint union of the nodes/edges/sources/receivers in the smaller networks.
  Let $\F$ and $\K$ be finite fields of characteristic $p$
  and $q$,
  for some distinct primes $p$ and $q$.
  Let $\Network$ be the disjoint union of
  the \charNetwork{$p$} and the \charNetwork{$q$}.
  Whenever node (respectively, edge and message) labels are repeated,
  add an arbitrary additional level of labeling each node (respectively, edge and message) to avoid repeated labels.
  Then, by Lemma~\ref{lem:char-p-network},
  the rate vector, in which the rates for 
  the \charNetwork{$p$}
  are all one
  and the rates for 
  the \charNetwork{$q$} are all zero,
  is linearly achievable over $\F$ but not over $\K$.
  Similarly,
  the rate vector in which the rates for 
  the \charNetwork{$q$}
  are all one
  and the rates for 
  the \charNetwork{$p$} are all zero
  is linearly achievable over $\K$ but not over $\F$.
  Thus the \linearrateregions{} of $\Network$ over $\F$ 
  and $\K$ do not contain one another.
\end{proof}

We can use a similar network construction
to show that
there is not necessarily a particular finite field
that can linearly achieve
all linearly achievable rate vectors.
In other words, 
there may not be a ``best'' field
for a particular network.
Let $p$ and $q$ be distinct primes,
and let $\Network$ be the disjoint union of
the \charNetwork{$p$} and the \charNetwork{$q$}.
Then, by a similar argument to the proof of Theorem~\ref{thm:char-not-contained},
there exists a rate vector that is only linearly achievable over fields of characteristic $p$,
and there exists another rate vector that is only linearly achievable over fields of characteristic $q$.
Thus there is no finite field which can linearly achieve both of these rate vectors.
A similar result can be obtained by taking the disjoint union of the Fano and non-Fano networks.

Theorem~\ref{thm:char-not-contained} demonstrates that for any two finite fields
of distinct characteristics,
there always exists some network whose \linearrateregions{} differ over the two fields.
In the following theorem, we show that the \linearrateregion{} of a network
over a field depends only on the characteristic of the field.
This contrasts with the scalar linear solvability of networks over fields,
since some networks can be scalar linearly solvable only over certain fields of a given characteristic.
\begin{theorem}
  Let $\F$ and $\K$ be finite fields.
  Then $\Char{\F} = \Char{\K}$
  if and only if for each network $\Network$, we have
  $\LinRegion{\Network}{\F} = \LinRegion{\Network}{\K}$.
  \label{thm:char-p}
\end{theorem}
\begin{proof}
  Let $r$ and $s$ be positive integers, $p$ a prime,
  and $\Network$ a network with $m$ messages.
  Then $\GF{p}$ is a subfield $\GF{p^s}$,
  which implies the identity mapping is an injective homomorphism from $\GF{p}$
  to $\GF{p^s}$.
  So
  \begin{align*}
    &\text{$\Network$ has a $(k_1,\dots,k_m,n)$ linear solution over $\GF{p^r}$}
    \\
    &\implies \;
    \text{$\Network$ has an $(rk_1,\dots,rk_m,rn)$ linear solution over $\GF{p}$}
      & & \Comment{Lemma~\ref{lem:char_p}}
      \\
    &\implies \;
    \text{$\Network$ has an $(rk_1,\dots,rk_m,rn)$ linear solution over $\GF{p^s}$}
      & & \Comment{Lemma~\ref{lem:ModHomomorphism}}.
  \end{align*}
  Both a $(k_1,\dots,k_m,n)$ linear solution
  and a $(rk_1,\dots,rk_m,rn)$ linear solution have 
  the rate vector $(k_1/n,\dots,k_m/n)$.
  Hence any rate vector that is linearly attainable over $\GF{p^r}$
  is also linearly attainable over $\GF{p^s}$
  (with possibly larger vector sizes).
  Similarly,
  any rate vector that is linearly attainable over $\GF{p^s}$
  is also linearly attainable over $\GF{p^r}$
  (with possibly larger vector sizes).
  Hence if $\Char{\F} = \Char{\K}$,
  then the \linearrateregions{} of any network over $\F$ and $\K$ are equal.
  The reverse direction follows from Theorem~\ref{thm:char-not-contained}.
\end{proof}

Immediately following Definition~\ref{def:Dominance},
we showed that
for any finite rings $S$ and $R$,
$$\Module{S}{S\otimes S^{op}} \text{ \fracDom{} } \Module{R}{R\otimes R^{op}}
\; \Longrightarrow \; 
\LinRegion{\Network}{S} \supseteq \LinRegion{\Network}{R} \text{ for every network } \Network.$$
Theorem~\ref{thm:char-p} can be used to show
the converse is not necessarily true.
There are numerous examples in the literature
(e.g., see \cite[\CommLemmaFieldNoDominance]{CZ-Commutative}, \cite{Sun-BaseField}, \cite{Sun-FieldSize})
of networks that are scalar linearly solvable over $\GF{p^r}$
but not over $\GF{p^s}$,
for some prime $p$ and some distinct positive integers $r$ and $s$.
In such cases, 
$\GF{p^s}$ does not \fracDom{} $\GF{p^r}$;
however, by Theorem~\ref{thm:char-p},
any network's \linearrateregion{} over either field is the same,
since both fields have characteristic $p$.

\begin{corollary}
  Let $\F$ and $\K$ be finite fields.
  Then $\Char{\F} = \Char{\K}$
  if and only if for each network $\Network$, we have
  $\UniLinCap{\Network}{\F} = \UniLinCap{\Network}{\K}$.
  \label{cor:char-p}
\end{corollary}
\begin{proof}
  This corollary is an immediate consequence of Theorem~\ref{thm:char-p}
  and Corollary~\ref{cor:char-unequal}.
\end{proof}

\newpage
\section{Linear Rate Regions over Rings} \label{sec:rings}

The following theorem demonstrates that if a network has a fractional linear solution over some module
and if $p$ is a prime that divides the alphabet size (i.e., the size of the group),
then the network must also have a fractional linear solution over every field of characteristic $p$
with the same rate vector and possibly larger vector sizes.

\begin{theorem}
  Let $\Module{G}{R}$ be a module
  and let $\F$ be a finite field whose characteristic divides $|G|$.
  For each network $\Network$ and each $k_1,\dots,k_m \ge 0 $ and $n \ge 1$
  such that $\Network$ has a $(k_1,\dots,k_m,n)$ linear solution over $\Module{G}{R}$,
  there exists a positive integer $t$
  such that $\Network$ has a $(tk_1,\dots,tk_m,tn)$ linear solution over $\F$.
  \label{thm:higher-dim-char-p}
\end{theorem}
\begin{proof}
  Let $p = \Char{\F}$.
  By Lemma~\ref{lem:min-module},
  there exists a finite field $\K$ of characteristic $p$ 
  and a positive integer $s$ such that
  $\Module{\K^s}{M_s(\K)}$ \fracDom{} $\Module{G}{R}$.
  Lemma~\ref{lem:vector_module}
  implies a network $\Network$
  with a $(k_1,\dots,k_m,n)$ linear solution over $\Module{\K^s}{M_s(\K)}$
  must also have
  an $(sk_1,\dots,sk_m,sn)$ linear solution over $\K$.
  Since $\F$ and $\K$ both have characteristic $p$,
  and since the rate vector $(k_1/n,\dots,k_m/n)$
  is linearly achievable for $\Network$ over $\K$,
  by Theorem~\ref{thm:char-p},
  the rate vector $(k_1/n,\dots,k_m/n)$ is also linearly achievable for $\Network$ over $\F$.
  Hence there exists a positive integer $t$
  such that $\Network$ has a 
  $(tk_1,\dots,tk_m,tn)$ linear solution over $\F$.
\end{proof}

We now prove one of our main results regarding \linearrateregions{} over rings.

\begin{theorem}
  If $R$ is a finite ring and $\F$ is a finite field whose characteristic divides $|R|$,
  then the \linearrateregion{} of any network over $R$
  is contained in the network's \linearrateregion{} over $\F$.
  \label{thm:ring-lin-cap}
\end{theorem}
\begin{proof}
  Let $R$ be a finite ring, 
  let $\Network$ be a network,
  and
  let $\F$ finite field whose characteristic divides $|R|$.
  A fractional two-sided linear solution over $R$
  is a fractional linear solution over the module
  $\Module{R}{R \otimes R^{op}}$,
  so by
  Theorem~\ref{thm:higher-dim-char-p},
  whenever $\Network$ has a fractional linear solution over $R$
  with a given rate vector,
  $\Network$ also has a fractional linear solution over $\F$
  with the same rate vector and possibly larger vector sizes.
  Hence, 
\begin{align*}
 & \{\mathbf{r} \in \Q^m \, : \, \mathbf{r} \text{ is linearly achievable for $\Network$ over } R \} \\
 &\ \ \  \subseteq
   \{\mathbf{r} \in \Q^m \, : \, \mathbf{r} \text{ is linearly achievable for $\Network$ over } \F \}. 
\end{align*}
\end{proof}

\begin{corollary}
  If $R$ is a finite ring and $\F$ is a finite field whose characteristic divides $|R|$,
  then the linear capacity of any network over $R$ is less than or equal to
  its linear capacity over $\F$.
  \label{cor:ring-lin-capacity}
\end{corollary}

In some cases,
the containment in Theorem~\ref{thm:ring-lin-cap}
(and the inequality in Corollary~\ref{cor:ring-lin-capacity}) 
is strict for some networks,
while in other cases,
there may be equality for all networks.
As an example,
by taking $\F = \GF{2}$ and $R = \Z_6$ in Theorem~\ref{thm:ring-lin-cap},
any network's \linearrateregion{} over $\GF{2}$
contains its \linearrateregion{} over $\Z_6$.
However, the linear capacity of the \charNetwork{$2$} 
is $1$ over the field $\GF{2}$
and is upper bounded by $6/7$ over the field $\GF{3}$
(see Lemma~\ref{lem:char-p-network}).
Since $3 = \Char{\GF{3}}$,
which divides $6 = |\Z_6|$,
by Corollary~\ref{cor:ring-lin-capacity},
the \charNetwork{$2$}'s linear capacity over $\Z_6$
is upper bounded by $6/7$.
This demonstrates that
the \linearrateregions{} of $R$ and $\F$
are not necessarily equal for all networks.

As another example,
by taking $\F = \GF{4}$ and $R = \Z_2[X]/\langle X^2 \rangle$ in Theorem~\ref{thm:ring-lin-cap},
any network's \linearrateregion{} over $\GF{4}$
contains its \linearrateregion{} over $\Z_2[X]/\langle X^2 \rangle$.
The field $\GF{2}$ is isomorphic to a subring of $\Z_2[X]/\langle X^2 \rangle$
(namely $\Z_2$),
so there is an injective homomorphism from $\GF{2}$ to $\Z_2[X]/\langle X^2 \rangle$,
which by Lemma~\ref{lem:ModHomomorphism},
implies any network's \linearrateregion{} over $\Z_2[X]/\langle X^2 \rangle$
contains its \linearrateregion{} over $\GF{2}$.
However, by Theorem~\ref{thm:char-p},
any network's \linearrateregions{} over $\GF{4}$ and $\GF{2}$ must be equal.
Thus the \linearrateregions{} of $\GF{4}$ and $\Z_2[X]/\langle X^2 \rangle$
are equal for all networks.
Precisely characterizing for which rings and fields
the \linearrateregions{} are equal for all networks
remains an open problem.

\subsection{Comparing Linear Capacities over Different Rings}

Determining the exact linear capacity and the \linearrateregion{} of the \charNetwork{$m$}
over each finite ring (or even each finite field) is also presently an open problem.
Another related open question is
for which finite rings $R$ and $S$ does there exist a network $\Network$
such that
$\UniLinCap{\Network}{R} > \UniLinCap{\Network}{S}$.
We have answered this second question in some select special cases:
\begin{itemize}
  \item In Theorem~\ref{thm:char-p}, 
  we showed that when $R$ and $S$ are finite fields,
  such a network exists if and only if the characteristics of $R$ and $S$ differ.

  \item In Theorem~\ref{thm:ring-lin-cap}, 
  we showed that when $S$ is a field whose characteristic divides $|R|$,
  no such network exists.
  This includes the special case where $|S|=|R|$.
\end{itemize}

\begin{corollary}
  Let $R$ and $S$ be finite rings.
  If some prime factor of $|S|$ is not a factor of $|R|$,
  then there exists a network $\Network$
  such that $\UniLinCap{\Network}{R} > \UniLinCap{\Network}{S}$.
  \label{cor:ring-capacity-r-s}
\end{corollary}
\begin{proof}
  Let $p$ divide $|S|$ but not $|R|$,
  and let $\Network$ denote the \charNetwork{$|R|$}.
  Then,
  \begin{align*}
  \UniLinCap{\Network}{S} 
    &\le \UniLinCap{\Network}{\GF{p}}  
      & \Comment{Theorem~\ref{thm:ring-lin-cap}}
      \\
    &\le 1-\frac{1}{2|R|+3}        
      & \Comment{$\NDiv{p}{|R|}$ and Lemma~\ref{lem:char-p-network}}
      \\
    &< 1 
    \\
    &\le \UniLinCap{\Network}{R} 
      & \Comment{$\Div{\Char{R}}{|R|}$}
  \end{align*}
  where the last inequality uses the fact that 
  $\Network$ must be scalar linearly solvable over $R$,
  since the characteristic of $R$ divides the size of $R$.
\end{proof}

Corollary~\ref{cor:ring-capacity-r-s} 
implies that if the sizes of two rings do not share the same set of prime factors,
then at least one of the rings induces a higher linear capacity than the other on some network.
As an example,
the \charNetwork{$6$} has a strictly larger linear capacity over the ring $\Z_6$
than over the field $\GF{25}$ of larger size.

Corollary~\ref{cor:ring-capacity-r-s}, in particular, implies that
for \textit{every} finite field and \textit{every} ring, whose sizes are relatively prime,
there is \textit{some} network 
for which the linear capacity of the network over the ring is strictly larger than the linear capacity  over the field.
In contrast, Theorem~\ref{thm:ring-lin-cap} shows that
for \textit{every} ring and \textit{every} network, 
there is \textit{some} field 
for which the linear capacity of the network over the ring is less than or equal to the linear capacity over the field.
These facts are succinctly summarized in the following theorem.

\begin{theorem}
  Let $\F$ be a finite field and $R$ be a finite ring.
  Then $|\F|$ and $|R|$ are relatively prime
  if and only if 
  there exists a network $\Network$ such that
  $\UniLinCap{\Network}{R} > \UniLinCap{\Network}{\F}$.
  \label{thm:ring-cap-iff}
\end{theorem}
\begin{proof}
  Let $p = \Char{\F}$.
  Then
  $|\F|$ and $|R|$ are relatively prime
  if and only if $\NDiv{p}{|R|}$.
  
  If $\NDiv{p}{|R|}$,
  then by Corollary~\ref{cor:ring-capacity-r-s},
  there exists a network $\Network$ such that
  $\UniLinCap{\Network}{R} > \UniLinCap{\Network}{\F}$.
  The converse is a restatement of Corollary~\ref{cor:ring-lin-capacity}.
\end{proof}

\subsection{Asymptotic Solvability}

We say that a network $\Network$ 
is \textit{asymptotically solvable over $\A$}
if for all $\epsilon \in (0,1)$,
the rate vector $$(1-\epsilon,\dots,1-\epsilon)$$
is contained in the network's \rateregion.
In other words, a uniform rate arbitrarily close to, or above, $1$ is attainable.
A network which is asymptotically solvable
but is not solvable was demonstrated in \cite{DFZ-Unachievability},
and non-linearly solvable networks were demonstrated
in \cite{CZ-NonLinear} and \cite{DFZ-Insufficiency}
that are not asymptotically linearly solvable over any finite field.
The following corollary demonstrates that such networks
are additionally not asymptotically linearly solvable over any module (or ring).

\begin{corollary}
  If a network is asymptotically linearly solvable over some module or ring,
  then it must be asymptotically linearly solvable over some finite field.
  \label{cor:module-asymptotic}
\end{corollary}
\begin{proof}
  Suppose a network $\Network$
  is asymptotically linearly solvable over some module $\Module{G}{R}$.
  By Theorem~\ref{thm:higher-dim-char-p},
  there exists a finite field $\F$ such that
  any rate vector that is linearly achievable over $\Module{G}{R}$
  must also be linearly achievable over $\F$.
  Hence $\Network$ is also asymptotically linearly solvable over $\F$.
  This also implies any network that is asymptotically linearly solvable over some ring
  must also be asymptotically linearly solvable over some field,
  since a fractional linear code over a ring is a special case of a fractional linear code over a module.
\end{proof}

\clearpage
\section{Concluding Remarks}
  \label{sec:conclusion}

Linear network codes over finite rings (and modules)
constitutes a much broader class of codes
than linear network codes over finite fields.
Linear codes over rings have many of the attractive properties
of linear codes over fields, including implementation complexity 
and possibly mathematical tractability.
We have demonstrated, however, that with respect to linear capacity and linear rate regions,
this broader class of codes
does \textit{not} offer an improvement over linear codes over fields.
This particularly contrasts with the network solvability problem
where we demonstrated certain cases where a ring alphabet
can offer scalar linear solutions when a field alphabet cannot.

\appendix

\section{Proofs of Lemmas in Section~\ref{sec:fractional}}

The proofs in this appendix
are results from \cite{CZ-NonCommutative}
that we include for completeness.

\subsection{Proof of Lemma~\ref{lem:same_ring} \cite[Lemma I.3]{CZ-NonCommutative}}
\LemSameRingProof

\subsection{Proof of Lemma~\ref{lem:faithful_module} \cite[Lemma II.6]{CZ-NonCommutative}}
\lemFaithfulModuleProof

\subsection{Proof of Lemma~\ref{lem:ModHomomorphism} \cite[Lemma I.6]{CZ-NonCommutative}}
\ProofOfLemModHomomorphism

\section*{Acknowledgment}

The authors wish to thank the anonymous reviewers
and the associate editor P. Sadeghi for some helpful suggestions.


\end{document}